\pdfoutput=1

\documentclass{article}
\usepackage{ijcai22}

\usepackage{times}
\usepackage{soul}
\usepackage{url}
\usepackage[hidelinks]{hyperref}
\usepackage[utf8]{inputenc}
\usepackage[small]{caption}
\usepackage{graphicx}
\usepackage{amsmath}
\usepackage{amsthm}
\usepackage{amssymb}
\usepackage{enumitem}
\usepackage{float}
\usepackage{booktabs}
\usepackage{algorithm}
\usepackage{algorithmic}
\newtheorem{theorem}{Theorem}
\newtheorem{corollary}{Corollary}
\newtheorem{definition}{Definition}[section]
\newtheorem{lemma}{Lemma}

\DeclareMathOperator{\sign}{sign}

\title{Manipulating Elections by Changing Voter Perceptions}

\author{
Junlin Wu
\and
Andrew Estornell\and
Lecheng Kong\And
Yevgeniy Vorobeychik
\affiliations
Washington University in St. Louis
\emails
\{junlin.wu, aestornell, jerry.kong, yvorobeychik\}@wustl.edu
}

\begin{document}
\maketitle

\begin{abstract}
The integrity of elections is central to democratic systems.
However, a myriad of malicious actors aspire to influence election outcomes for financial or political benefit.
A common means to such ends is by manipulating perceptions of the voting public about select candidates, for example, through misinformation.
We present a formal model of the impact of perception manipulation on election outcomes in the framework of spatial voting theory, in which the preferences of voters over candidates are generated based on their relative distance in the space of issues.
We show that controlling elections in this model is, in general, NP-hard, 
whether issues are binary or real-valued.
However, we demonstrate that critical to intractability is the diversity of opinions on issues exhibited by the voting public.
When voter views lack diversity, and we can instead group them into a small 
number of categories---for example, as a result of political polarization---the election control problem can be solved in polynomial time in the number of issues and candidates
for arbitrary scoring rules.
\end{abstract}

\section{Introduction}

Elections are among the core functional elements of democratic systems.
Consequently, there is broad consensus that their integrity is among the top democratic priorities.
However, malicious actors may attempt to subvert elections for their own means, whether financial or political~\cite{Caldwell16,Harper19,Khetani-Shah19}.
A common approach for manipulating elections is by spreading false information about select candidates, an extreme example of which is the infamous ``Pizzagate'' campaign targeting Hillary Clinton in the 2016 U.S. presidential election~\cite{Robb17}.
Less extreme, but far more common, is the spread of misinformation about the positions of candidates on specific issues, such as taxation and debt.

The issue of election vulnerability to malicious manipulation has been studied in the computational social choice literature from a computational complexity perspective under the terms \emph{election control} (when the election structure itself is manipulated)~\cite{Bartholdi92,Hemaspaandra07,Chen17} and \emph{bribery} (when manipulation is through changing voter preferences over candidates)~\cite{Bredereck16,Faliszewski16}.

The traditional study of election control  takes voter preferences as given, while considerations of bribery investigate direct manipulations of preference rankings of individual voters.
However, neither is a natural model of the impact of misinformation \emph{about particular issues} on the perceptions of candidates by the voting public.
To address this gap, we propose a new model of election control in the \emph{spatial voting theory} framework.
Spatial voting theory explicitly captures voter and candidate positions on issues, with voter preferences over candidates determined by their relative distance in issue space~\cite{Anshelevich16,Anshelevich18,Enelow84}.
In our model of election manipulation, a malicious party can change voter perceptions of a target candidate on issues, subject to a budget constraint (more precisely, we constrain the $l_p$ norm of the manipulation to be below a specified bound).\footnote{This model can also be viewed as an example of bribery, in the sense that the manipulation affects voter preference rankings over candidates.  Our use of the term \emph{election control} is general, referring to any setting in which a malicious party wishes to subvert an election, whatever means they use for doing so.} We consider both constructive control, where the malicious goal is to cause the target candidate to win, and destructive control, in which the goal is to cause the target candidate to lose.

We show that when the issues are binary-valued, the problem is hard even with two candidates, for both forms of control and for any $l_p$ norm with integer $1 \le p < \infty$ used to measure distance in issue space. 
When issues are real-valued, however, the conclusions for constructive and destructive control differ slightly. For destructive control, the problem is hard even with two candidates.
For constructive control, we show hardness for plurality elections when the number of candidates is arbitrary, for $l_p$ norm with integer $1 < p \le \infty$. 
However, if there are only two candidates and we measure distance using $l_\infty$ norm, the control problem can be solved in polynomial time.
Furthermore, we show that if we restrict either the number of issues or the number of voters to be bounded by a constant, all control problems become tractable, whether issues are binary (for arbitrary $l_p$ norm) or real-valued (for $l_2$ and $l_\infty$), for arbitrary scoring rules used to determine election outcomes.
Moreover, we show that the tractability generalizes even when the number of voters is arbitrary, but their opinions on issues are limited to only a constant (that is, small) number of options.

These seemingly highly technical results offer a broader insight: vulnerability of elections to malicious manipulation of voter perceptions hinges on the extent to which voters exhibit a high diversity of political views.
When this is the case, elections are highly resistant to manipulation.
However, \emph{when voters are Balkanized into a small number of groups with effective uniformity of views within each, for example, due to political polarization, elections become easy to manipulate through misinformation}.

Our model of election control is related to several recent studies of election control in the spatial voting theory framework~\cite{Lu19,Estornell20}.
However, the means of manipulation in this closely related work is changing the relative importance of issues to voters, whereas our focus is on changing voter perceptions of candidates.

\paragraph{Related Work.}
The study of election control was initiated by \citeauthor{Bartholdi92}~\shortcite{Bartholdi92}, who studied the impact of adversarially adding, deleting, or partitioning candidates or voters on election outcomes in the constructive control framework.
Numerous follow-up efforts extended this analysis in a number of directions, such as destructive control~\cite{Hemaspaandra07}, a variety of voting rules and settings~\cite{Menton12,Erdelyi15,Chen17}, as well as when voter preferences can be modified (commonly called the \emph{bribery problem}~\cite{Bredereck16,Faliszewski16} or \emph{optimal lobbying}~\cite{Christian2007,Binkele2014}).

In most election control settings voter preferences are specified directly as preference rankings over the candidates.
An alternative approach based on spatial theory of voting, specifies voter and candidate positions on issues, with preference rankings then induced from relative distances between voter and candidate positions~\cite{Davis68,Enelow84,Anshelevich16,Anshelevich18}.
\citeauthor{Lu19}~\shortcite{Lu19} were the first to investigate election control within the spatial theory voting model, with the adversary's ability restricted to selecting a subset of issues that become the focus of voting preferences.
\citeauthor{Estornell20}~\shortcite{Estornell20} study a variation in which an adversary can modify the relative importance of issues in determining voter preferences over candidates.
Both are distinct from our model in which the adversary modifies not the importance of issues, but the \emph{perceptions} of a particular candidate by the voters.

Several models of election control are also motivated by the spread of misinformation about candidates on social networks~\cite{Wilder18,Castiglioni20}.
However, these focus on stochastic spread of misinformation in the social influence modeling framework~\cite{Kempe03}, but use the conventional model of elections in which voter preferences are directly specified, with misinformation having a direct impact on a target candidate's relative ranking for a given voter, rather than an indirect impact stemming from the change in perceived positions on issues, as in our model.

\section{Preliminaries}

We consider an election with a set of $n$ candidates $\mathcal{C}=\{c_1,\ldots,c_{n}\}$ and $m$ voters $\mathcal{V}=\{v_1,\ldots,v_m\}$.
Following spatial voting theory~\cite{Enelow84}, we associate each candidate and voter with a $d$-dimensional vector corresponding to their positions (opinions) on issues, that is, $c_i, v_j \in \mathcal{I} \subseteq \mathbb{R}^d$.
Each voter $v_j$ ranks candidates in $\mathcal{C}$ according to their $l_p$ distance from $v_j$, $\|v_j - c_i\|_p$, with $1 \le p \le \infty$ an integer; the closest candidate is ranked 1, and the farthest is ranked $n$ in the list of $v_i$'s preferences.
If not mentioned, the parameters of the problem (e.g., $|\mathcal{V}|$, $|\mathcal{C}|$ and $d$) are arbitrary.
We assume that there are no ties. 

In our election control problem, the adversary has a target candidate whose voter perceptions they can manipulate.
Without loss of generality, let $c_1$ be the target candidate.
We assume that the adversary can change the perception of $c_1$ into $\tilde{c}_1$, subject to the constraint that $||\tilde{c}_1 - c_1||_p \leq \epsilon$ for $\epsilon > 0$.
This ``budget'' constraint is natural: for example, if the means for changing perceptions is social media misinformation, the change to perception is likely gradual, and one cannot target arbitrary subsets of issues with an arbitrarily large stream of malicious content.
We consider two types of control: constructive, in which the adversary's goal is for $c_1$ to win the election, and destructive, where the goal is for $c_1$ to lose.
While we assume no ties in the \emph{actual} preference rankings, ties can arise due to adversarial activities; in that case, we always break ties in the adversary's favor.

We consider election control problems for arbitrary \emph{scoring rules}.
In scoring rules, each candidate $c_i$ ranked $t_{ij}$ by voter $v_j$ receives a score $f(t_{ij})$, where $f: [n] \rightarrow \mathbb{R}$ is a non-increasing function.
$c_i$ then receives a total score $s_i = \sum_j f(t_{ij})$ from all voters, and the candidate with the highest score $s_i$ wins the election.
Many common voting rules are positional, such as plurality ($f(1) = 1$, and $f(t) = 0$ for $t \ne 1$), veto ($f(n)=0$ and $f(t) = 1$ for $t \ne n$), Borda ($f(t) = n-t$), and $k$-approval (for some $1 \le k \le n$, $f(t) = 1$ for all $t \le k$, and $f(t) = 0$ for $t > k$). Note that plurality is a special case of $k$-approval with $k=1$.

We study the problem both when issues are binary, i.e., $\mathcal{I} = \{0,1\}^d$, salient if issues are framed in the form of yes-no questions, such as ``do you support leaving the European Union?'', and when issues are real-valued ($\mathcal{I} = \mathbb{R}^d$).

\section{Binary-Valued Issues}
\label{S:bvpm}

We begin by studying a special case of our problem in which the issues are binary, that is, $\mathcal{I} = \{0,1\}^d$, a variant we call
Binary Value Perception Manipulation (\textsc{BVPM}).
\begin{definition}[\textsc{BVPM}]
Given a set of candidates $\mathcal{C}$, voters $\mathcal{V}$, and $d$ issues, is there a $\tilde{c}_1\in \mathcal{I}= \{0,1\}^d$ where ${||\tilde{c}_1 - c_1||_p \leq \epsilon}$ for $\epsilon > 0$ such that $\tilde{c}_1$ wins the election?
\end{definition}
Note that for $l_\infty$ the problem is trivial: either $\epsilon \geq 1$, in which case we can set $c_1$ to match any currently winning candidate (and $c_1$ wins by best-case tiebreaking), or $\epsilon < 1$ and we cannot change $c_1$.
Thus all the results in this section are for an arbitrary $l_p$ norm for $p$ integer, $1 \leq p < \infty$. 
Without loss of generality (since the label of 1 or 0 for each issue is arbitrary), we assume that $c_1$ takes a position labeled as 1 for each issue, i.e., $c_1 = [1,\ldots,1]$.

We begin by showing that even with 2 candidates and majority voting the \textsc{BVPM} problem is NP-complete. We reduce from Binary Issue Selection Control (\textsc{BISC}), shown by \citeauthor{Lu19}~\shortcite{Lu19} to be NP-Complete with best-case tie-breaking even when $|\mathcal{C}|=2$.
\begin{definition}[\textsc{BISC}~\cite{Lu19}]
Given a set of candidates $\mathcal{C}$, voters $\mathcal{V}$, and $d$ issues, is there a nonempty subset of binary issues $S \subseteq [1:d]$ such that a target candidate $c_1$ wins the plurality election?
\end{definition}
\begin{theorem}
\textsc{BVPM} is NP-complete for constructive and destructive control even with 2 candidates and majority voting.
\label{BVPMmajority}
\end{theorem}
\vspace{-1.5em}
\begin{proof}
It is easy to check if a given $\tilde{c}_1$ wins the election; thus, \textsc{BVPM} is in NP.
We now show hardness for constructive control by reduction from \textsc{BISC}.
Let $\epsilon=(d-1)^{1/p}$, target candidate $c_1 = [1,1,\ldots,1]$ and rival candidate $c_2 = [0,0,\ldots,0]$. 
Let the voter set of \textsc{BVPM} be the same as the one in \textsc{BISC}.

Suppose 2-candidate \textsc{BISC} has a solution  ${S \subseteq [1:d]}, S \neq \emptyset$ that will let ${c}_1$ win the election. For any voter $v_j \in \mathcal{V}, j \in [m]$ that votes for target candidate in \textsc{BISC}, by the problem definition of \textsc{BISC}, we have $\sum_{k\in S} {|{c}_{1,k}-v_{j,k}|^p }\leq \sum_{k\in S} |c_{2,k}-v_{j,k}|^p$. We set $\tilde{c}_{1,k} = 1$ if $k \in S$; $\tilde{c}_{1,k} = 0$ if $k \in S^c$. Since $S \neq \emptyset$, ${|{S}^c|\leq d-1}$, $\tilde{c}_1$ is within the budget constraint. For $k \in S$, we have ${\sum_{k\in S} |\tilde{c}_{1,k}-v_{j,k}|^p} = {\sum_{k\in S} |{c}_{1,k}-v_{j,k}|^p} \leq {\sum_{k\in S} |c_{2,k}-v_{j,k}|^p}$. For $k \in {S}^c$, since ${\tilde{c}_{1,k} =c_{2,k}= 0}$, we then have $|\tilde{c}_{1,k}-v_{j,k}|^p = |c_{2,k}-v_{j,k}|^p$. Thus we have ${\sum_{k=1}^d |\tilde{c}_{1,k}-v_{j,k}|^p} \leq {\sum_{k=1}^d |c_{2,k}-v_{j,k}|^p}$. %
This means voter $v_j$ votes for $\tilde{c}_1$ in \textsc{BVPM} $({||\tilde{c}_1-v_j||_p\leq ||c_2-v_j||_p})$. $\tilde{c}_1$ wins the election and is a solution to two candidates \textsc{BVPM}.

If two candidates \textsc{BVPM} has a solution $\tilde{c}_1$ that wins the election, by the problem definition of BVPM, we have ${\sum_{k=1}^d |\tilde{c}_{1,k}-v_{j,k}|^p} \leq {\sum_{k=1}^d |c_{2,k}-v_{j,k}|^p}$. We let ${S=\{k\in[d]\mid \tilde{c}_{1,k} = 1\}}$. Since $\tilde{c}_{1,k} =c_{2,k}= 0$, ${k\in {S}^c}$, we have $\sum_{k \in S^c}|\tilde{c}_{1,k}-v_{j,k}|^p = {\sum_{k\in S^c}|c_{2,k}-v_{j,k}|^p}$; since $\tilde{c}_{1,k}= {c}_{1,k} = 1, k\in S$, we then have $\sum_{k\in S} |{c}_{1,k}-v_{j,k}|^p={\sum_{k\in S} |\tilde{c}_{1,k}-v_{j,k}|^p} \leq \sum_{k\in S} |c_{2,k}-v_{j,k}|^p$. Due to the budget constraint $\epsilon = (d-1)^{1/p}$, we must have $S\neq \emptyset$, which satisfies the \textsc{BISC} solution requirement. $S$ is a solution set to two candidates \textsc{BISC}.

For destructive control, the same argument applies after switching the positions of $c_1$ and $c_2$.
\end{proof}

While \textsc{BVPM} is hard in general, we next show that the problem is tractable for a constant number of voters.\footnote{Note that this is trivial for a constant number of binary issues.}
While at first glance a constant number of voters seems an impractical restriction, \emph{we subsequently show that this result offers real insight even when the number of voters is arbitrary}.

\begin{theorem}
When the number of voters is constant, \textsc{BVPM} can be solved in polynomial time  for arbitrary scoring rules, for both constructive and destructive control.
\label{BVPMconstantvoter}
\end{theorem}
\begin{proof}[Proof sketch]
Given an issue $j$, let $\mathbf{v}^j$ be a vector corresponding to the position of each voter on issue $j$.
The key idea is that when the number of voters $m$ is a constant, there is also a constant number $l \le 2^m$ of \emph{issue equivalence sets}, where an equivalence set $I$ is a set of issues with identical $\mathbf{v}^j$ (consequently, each issue in $I$ has an identical and interchangeable impact on the election outcome).
Since the number of issues of each equivalence set we are allowed to flip is at most $\left \lfloor{\epsilon^{p}}\right \rfloor \leq d$, it is direct that we can exhaustively enumerate all $O(d^{2^m})$ possibilities for arbitrary $\epsilon$, which is polynomial since $m$ is constant. 
\end{proof}

While there is a simple poly-time algorithm for solving \textsc{BVPM}, we can actually considerably improve on its time complexity by leveraging additional problem structure.
We begin with the constructive control case for arbitrary scoring rules. A key feature of arbitrary scoring rules is that as long as $c_1$ receives (one of) the highest scores, $c_1$ wins the election. 
Since the distances between each voter $v_j$ and all candidates $c_i$ other than $c_1$ are fixed,  the relative rankings of candidates $c_i~(i \geq 2)$ w.r.t. voter $v_j, j \in[m]$ are fixed as well. 
Given the rankings of candidates $c_i ~(i \geq 2)$ w.r.t. to voter $v_j$ from closest to furthest as $c_{i_1},\ldots,c_{i_{n-1}}$, we can enumerate all scenarios of insertion positions of $c_1$ in this sequence. 
We denote the final ranking of $c_1$ after insertion w.r.t. voter $v_j$ as $r_j$, meaning $c_1$ will receive a score of $f(r_j)$ from $v_j, j\in[m]$. By going through all scenarios of $c_1$ getting a final ranking position of $(r_1,r_2,\ldots,r_m)$, $r_j \in [n], \forall j \in[m]$, which corresponds to $c_1$ getting a score of $(f(r_1),f(r_2),\ldots,f(r_m))$, we cover all the scenarios of ${c}_1$ winning. 
As shown in Lemma~\ref{constructiverank}, this is equivalent to enumerating all scenarios of ${c}_1$ getting a ranking position \textit{higher} than $r_j\in[n]$ w.r.t. $v_j$ for $j\in[m]$.
The missing proofs of this and other results are provided in the Supplement.
\begin{lemma}
\label{constructiverank}
For constructive control, if a ranking position $(r_1,r_2,\ldots,r_m)$ is feasible for ${c}_1$ under the budget constraint and lets ${c}_1$ win the election, then for a ranking position $(r'_1,r'_2,\ldots,r'_m)$ that is feasible with $r'_j\leq r_j, \forall j\in [m]$, it will also let ${c}_1$ win the election.
\end{lemma}

This means by enumerating all scenarios of ${c}_1$ getting a ranking position \textit{higher} than $(r_1,r_2,\ldots,r_m)$, which corresponds to $c_1$ getting a score of \textit{at least} $(f(r_1),f(r_2),\ldots,f(r_m))$, we cover all the possible scenarios of ${c}_1$ winning.
In fact, we can further simplify the calculation by only enumerating some ranking positions each corresponding to a \textit{unique} $f$ score value. 
Given an arbitrary scoring function $f$ that has $r$ unique values $({|f_{\text{uniq}}|=r})$, since $f$ is non-increasing, we can partition the domain of $f$ by ${0 = s_0 < s_1< \dots < s_r = n}$, so that $\{f(k)\}_{k=s_i+1}^{s_{i+1}}$ have the same value, ${i\in\{0,1,\ldots,r-1\}}$.
This means $s_{i+1}$ is the lowest ranking position that corresponds to score $f(s_{i+1})$ and $\{f(s_i)\}_{i=1}^r$ contains all the unique values of $f$.

\begin{lemma}
\label{constructiverep}
For constructive control, by enumerating all scenarios of ${c}_1$ getting a ranking position higher than $t_j$ w.r.t. voter $v_j$ for $t_j \in\{s_1,\ldots,s_r\},\forall j\in [m]$, we cover all the possible scenarios of ${c}_1$ winning.
\end{lemma}

Next we solve the problem for each $(t_1,\ldots,t_m)$ scenario with $t_j \in\{s_1,\ldots,s_r\}, \forall j \in [m]$. For each voter $v_j, {j \in[m]}$, we rank candidates $c_i ~(i \geq2)$ by their distances to voter $v_j$ from closest to furthest, and use $d_j^{t_j}$ to denote the distance between $v_j$ and the candidate ranked $t_j$-th closest to it. Since the tie breaks in the adversary's favor, as long as $c_1$'s distance to $v_j$ is no more than $d_j^{t_j}$, $c_1$ will receive a score of at least $f(t_j)$ from $v_j$. Notice that since the rankings of $c_i~(i\geq 2)$ do not include $c_1$, only $d_j^{t_j}$ for $1\leq t_j\leq n-1$ are properly defined. For $t_{j} = s_r= n$, we let $d_j^{n} = +\infty$, since $c_1$ is guaranteed to get at least the lowest score $f(n)$. 

For each scenario, the problem can be represented as the following integer linear constraint problem:
\begin{subequations}
\begin{align}
    &x_i \leq \min\{b_i,\left \lfloor{\epsilon^{p}}\right \rfloor\}, \quad 1\leq i\leq 2^m\\
    &\sum_{i=1}^{2^m}z_{ij} x_i +(d_{j}^0)^p\leq (d_j^{t_{j}})^p,\quad 1\leq j \leq m \label{dist}\\
    & \sum_{i=1}^{2^m} x_i \leq \left \lfloor{\epsilon^{p}}\right \rfloor \quad x_i \in \mathbb{Z}^+,
\end{align}
\end{subequations}
where $x_i$ is the number of issues in an issue equivalence class $i$ that we want to flip to $0$, $b_i$ is the size of the \mbox{$i$-th} issue equivalence class (where $\sum_i b_i = d$),
$d_{j}^{0}$ is the original distance between target candidate $c_1$ and voter $v_j$. 
$z_{ij} \in \{-1,+1\}$ is the sign of impact of flipping the issue: $z_{ij} = +1$ if previously the $j$-th voter in any issue in the $i$-th equivalence class is $1$ (since flipping the $c_1$ to $0$ will increase the distance) and $z_{ij} = -1$ if previously it is $0$, since flip the target candidate to $0$ will decrease the distance.
Since the size of the input of this integer feasibility problem is $O(\log(d))$ \emph{and the number of variables is constant}, it can be solved in time $O(\log(d))$~\cite[Theorem 2.8.1]{Lokshtanov09}. 

The total number of times we need to run the ILP and check whether $\tilde{c}_1$ wins the election is bounded by the number of unique $(t_1,\ldots,t_m)$ scenarios, which is ${|f_{\text{uniq}}|}^m$. 
For an arbitrary scoring function $f$, the time complexity for calculating voter-candidate distances is $O(n d)$, ranking distances takes $O(n\log(n))$ time. The calculation of the issue equivalence sets takes $O(d)$ time. For each scenario, solving the ILP takes $O(\log (d))$ time; calculating the distance between $\tilde{c}_1$ and all voters takes $O(d)$ time; checking whether $\tilde{c}_1$ wins the election takes $O(n)$ time. The total time complexity of the algorithm is ${O(n(d + \log(n)) + |f_{\text{uniq}}|^m (d+n))}$.

If, in addition, $|f_{\text{uniq}}|=r$ is constant (e.g., for plurality), then the number of scenarios ${|f_{\text{uniq}}|}^m$ is constant. 
Moreover, we do not need to do a total sort for the distances. Through finding the $t_j$-th order statistics and Quicksort Partition, the total time complexity of the algorithm is $O(n d)$.

In Supplement~\ref{S:supp_binary_destructive} we present a similar analysis and algorithm as above for destructive control.
In either case, the complexity is linear in the dimension $d$ of the issue space.

As noted earlier, considering a constant number of voters may seem unrealistic.
However, we note that these algorithms are straightforward to generalize to a setting with an \emph{arbitrary number of voters}, but in which the positions of voters on issues, $v_j$, can only take on values from a small collection of possibilities $Q$ (that is, $|Q|$ is a constant, and for each voter $j$, $v_j \in Q$).
Specifically, the only change is to calculate the \textit{weighted} final score in each case above, where the weight for each distinct voter position (opinion) type $q \in Q$ is the number of voters $j$ with position $v_j = q$.
This is expressed in the following corollary.
\begin{corollary}
\textsc{BVPM} can be solved in polynomial time when the number of distinct voter opinions is constant for constructive and destructive control for arbitrary scoring rules.
\end{corollary}
The key insight that our results offer is that
\emph{the complexity of manipulating elections by changing voter perceptions about candidates hinges on the diversity of opinions among voters}.
In particular, when voters hold a broad diversity of views, manipulation is intractable; if, in contrast, voters are siloed into a relatively small collection of echo chambers, subverting elections becomes easy. %
Below, we show that this observation extends to real-valued issues.

\section{Real-Valued Issues}
Next, we turn to Real Value Perception Manipulation (\textsc{RVPM}), the problem identical to \textsc{BVPM} except that now the issue space $\mathcal{I}$ is real-valued.

\subsection{Hardness Results}
We begin by showing that nearly every variant of \textsc{RVPM} is, in general, computationally intractable.
First, we show that the election control problem is hard under $l_p$ norm with integer $1<p\leq \infty$ under destructive control even with 2 candidates and majority voting, and constructive control even for plurality voting. 
Nonetheless, we show that constructive control with $l_\infty$ norm and only two candidates is in P.

Our hardness result for destructive control %
uses a reduction from \textsc{3-SAT}, the proofs are deferred to the supplement (\ref{S:rvpm_constructive_hardness_lp} and \ref{S:rvpm_constructive_hardness_linf}).

\begin{theorem}
\label{rvpm_des_lp}
The destructive control variant of \textsc{RVPM} is NP-complete under $l_p$ norm for integer $1 < p \leq \infty$ even with two candidates and majority voting.
\end{theorem}

The following theorem (proved in Supplements~\ref{S:supp_cc_manyc_lp} and~\ref{S:supp_cc_manyc}) shows that constructive control is also hard.
\begin{theorem}
\label{T:cc_real_manyc_lp}
The constructive control variant of \textsc{RVPM} under $l_p$ norm for integer $1< p \leq \infty$ is NP-complete for plurality voting.
\end{theorem}

Note that Theorem~\ref{T:cc_real_manyc_lp} is stated for an arbitrary number of candidates.
For two candidates and $l_\infty$ norm, however, constructive variant of \textsc{RVPM} is easy.
\begin{theorem}
\label{T:rvpm_cc_2}
The constructive control $l_\infty$ norm variant of \textsc{RVPM} with 2 candidates can be solved in time $O(m d)$ for arbitrary scoring rules.
\end{theorem}
The proof is provided in Supplement~\ref{S:supp_rvpm_cc_2}.
Note that our results do not resolve the question of constructive control with $p < \infty$; we leave it as an open question.

Next, we consider two restrictions of \textsc{RVPM}: 1) assuming a constant number of issues, and 2) assuming a constant number of distinct voter opinions.
In all these restricted cases, we show how to solve \textsc{RVPM} in polynomial time for $l_2$ and $l_\infty$ norm.
We leave the problem open for arbitrary $l_p$ norms.

\subsection{Constant Number of Issues}
\label{S:rvpm_const_issues}

When the number of issues is constant, we show that RVPM is tractable for $l_2$ and $l_\infty$ norms (our focus on these two norms follows the precedent from prior literature~\cite{Crama95,Crama97}).
Note that unlike with binary issues, tractability of \textsc{RVPM} with a constant number of issues is \emph{non-trivial} since the issue space is continuous and cannot be exhaustively searched in finite time.

\subsubsection{Constructive Control}

We start by studying the problem of constructive control, with \textsc{RVPM} in that case closely related to the well-known \emph{product positioning} and \emph{ball intersection problems}~\cite{Crama95}.
 The goal in product positioning is to find $x \in \mathbb{R}^m$ that maximizes the number of consumers for whom $x$ is closer (in $l_p$) to their ideal product than any of the competitors.
The ball intersection problem aims to maximize the weighted sum of $l_p$ balls to which $x$ belongs.
Neither exactly captures our problem given the presence of the attacker budget constraint and different scoring scenarios, but both are useful tools in constructing the algorithms for our problem below. 

\begin{theorem}
\label{T:rvpm_c_ci}
\textsc{RVPM} can be solved in polynomial time under $l_2$ norm when the number of issues is constant for constructive control for arbitrary scoring rules.
\end{theorem}
\begin{proof}

We first convert our problem into the ball intersection problem.
Let $B_j^{t_j}$ be the ball that corresponds to voter $v_j$, with radius $d_{j}^{t_j}$ as defined in the constructive control variant of \textsc{BVPM}, that is, 
\[
B_{j}^{t_j} = \{\tilde{c}_1 \in \mathbb{R}^d\mid ||\tilde{c}_1 - v_j||_2 \leq d_{j}^{t_j}\}.
\]
Similarly, we define the candidate budget ball 
\[
B_{c} = \{\tilde{c}_1 \in \mathbb{R}^d \mid||\tilde{c}_1 - c_1||_2 \leq \epsilon\}.
\]
Since the tie breaks in the adversary's favor, $\tilde{c}_1$ will receive a score of at least $f(t_j)$ from voter $v_j$ iff $\tilde{c}_1$ falls within $B_j^{t_j}$ and $B_c$.
According to Lemma \ref{constructiverep} (which also applies when issues are real-valued), by  finding a representative point $\tilde{c}_1$ %
within ${\{B_c\} \cup (\bigcup_{j=1}^m \{B_j^{t_j}\})}$ for all scenarios of $(t_1,\ldots,t_m)$ with $t_{j}\in\{s_1,\ldots,s_r\}$ (partition of the domain of $f$ based on unique values as defined in the constructive control variant of BVPM), $\forall j \in [m]$, we cover all the scenarios of ${c}_1$ winning. %
We can now directly apply the ball intersection algorithm by \citeauthor{Crama95}~\shortcite{Crama95} for $l_2$ and constant $d$ to our problem only once for ${\{B_c\}\cup (\bigcup_{j=1}^m \bigcup_{l=1}^r \{B_j^{s_l}\})}$, the resulting set $P$ (representative points of intersections) contains a representative point for all the scenarios.
We check for each point in $P$ whether it is within $B_c$ and wins the election. 
The time complexity of the algorithm is exponential only in $d$.%
\end{proof}

A similar result, based on a similar connection to box intersection, can be obtained for the $l_\infty$ norm. 

\begin{theorem}
\label{T:rvpm_linf_cd}
\textsc{RVPM} can be solved in polynomial time under $l_\infty$ norm when the number of issues is constant for constructive control for arbitrary scoring rules.
\end{theorem}

\subsubsection{Destructive Control}

Next we study the problem under destructive control.
Define an open voter ball corresponding to voter $v_j$ with radius $d_{j}^{t_j}$ under $l_2$ norm
as \({\mathring{B}_{j}^{t_j}=\{\tilde{c}_1 \in \mathbb{R}^d\mid ||\tilde{c}_1 - v_j||_2 < d_{j}^{t_j}\},}\)
and recall that $d_{j}^{t_j}$ is the distance between $v_j$ and the candidate ranked $t_j$-th closest to it.
The next lemma, proved in the Supplement, provides an important building block.
\begin{lemma}
     Given $k$ open balls $\{\mathring{B}_1,\ldots,\mathring{B}_k\}$ and a closed ball $B_c$, let $P$ be the \emph{representative points of intersections} (defined in \citeauthor{Crama95}~\shortcite{Crama95}) w.r.t. $\{B_c, B_1,\ldots,B_k\}$, where $B_j$ is the closed ball corresponds to $\mathring{B}_j,~\forall {j\in[k]}$. For any given family of open balls
     $\{\mathring{B}_{{i_1}},\ldots,\mathring{B}_{{i_r}}\}\subseteq \{\mathring{B}_1,\ldots,\mathring{B}_k\}$, let ${P' = \{x\mid x \notin \mathring{B}_{j},} \forall j \in \{i_1,\ldots,i_r\}, x \in B_c\}$. If ${P' \neq \emptyset}$, then $P\cap P'\neq \emptyset$.%
    \label{avoid_ball}
\end{lemma}

Next, we show that when the number of issues is constant, the destructive control variant of \textsc{RVPM} is tractable.
\begin{theorem}
\label{T:rvpm_d_ci}
When the number of issues is constant, the destructive control variant of \textsc{RVPM} can be solved in polynomial time under $l_2$ norm for arbitrary scoring rules.
\end{theorem}
\begin{proof}
Since the tie breaks in the adversary's favor, within $B_c$, $\tilde{c}_1$ will get a score of no more than $f(t_j)$ from voter $v_j$ iff it falls outside of $\mathring{B}_j^{t_j}$. Similar to Theorem \ref{T:rvpm_c_ci}, we cover all the scenarios of ${c}_1$ losing by finding the set that contains a representative point within $B_c$ that falls outside of $\bigcup_{j=1}^m \{\mathring{B}_j^{t_j}\}$ for all scenarios of $(t_1,\ldots,t_m)$ with $t_{j}\in\{s_1,\ldots,s_r\}$ (partition of the domain of $f$ based on unique values as defined in the destructive control variant of BVPM), $\forall j \in [m]$. Lemma \ref{avoid_ball} shows us that the set $P$ (representative points of intersections) for the family of balls $\{B_c\}\cup (\bigcup_{j=1}^m \bigcup_{l=1}^r \{B_j^{s_l}\})$ contains a representative point for all the scenarios.
The problem could be solved in polynomial time with minor modifications to the 
algorithm in Theorem \ref{T:rvpm_c_ci}. %
\end{proof}

The destructive control problem with $l_\infty$ norm involves solving a non-convex feasibility problem.
The next lemma shows that the problem of relevance can nevertheless be solved in polynomial time.
\begin{lemma}
For the feasibility problem 
\begin{subequations}
\begin{align}
    &||\tilde{y}-y||_\infty \leq \epsilon \label{budget}\\
    &||\tilde{y}-a_i||_\infty \geq b_i, \quad i \in [k] \label{cube_c} \\
    & \epsilon>0, b_i > 0\quad \tilde{y}, y, a_i \in \mathbb{R}^d,
\end{align}
\end{subequations}
if $\tilde{y}\in \mathbb{R}^d$ satisfy constraint (\ref{budget}) and $\cup_{j=1}^d S_j(\tilde{y}_j) = [k]$, where ${S_j(\tilde{y}_j) = \{i\in[k]\mid  |\tilde{y}_j-a_{i,j}| \geq b_i\}}$, $\forall j\in[d]$, then $\tilde{y}$ is a solution to the feasibility problem. Moreover, for all $j \in [d]$, let set ${P}_j = \bigcup_{i=1}^k({\{-b_i+a_{i,j}, b_i + a_{i,j}\}}\cap {[-\epsilon +y_j,\epsilon +y_j]})$, or ${P}_j = \{y_j\}$ if the set is empty, then ${P=\{p\in \mathbb{R}^d\mid p_j\in {P}_j, \forall j \in [d]\}}$ contains a representative solution point $\tilde{y}$ to the problem. 
\label{non_convex}
\end{lemma}
We use Lemma~\ref{non_convex} to show that the destructive control variant of \textsc{RVPM} is also tractable for the $l_\infty$ norm.
\begin{theorem}
\textsc{RVPM} can be solved in polynomial time under $l_\infty$ norm when the number of issues is constant for destructive control for arbitrary scoring rules.
\end{theorem}
\begin{proof}
We solve the problem by using Lemma \ref{non_convex} to find the set $P$ which contains a representative solution point for all the scenarios of $(t_1,\ldots,t_m)$ (defined as in Theorem~\ref{T:rvpm_d_ci}). Each scenario represents finding a $\tilde{c}_1$ within the budget constraint that gets a score of no more than $f(t_j)$ from voter $v_j, j \in [m]$:%
\begin{subequations}
\label{E:lfp}
\begin{align}
    &||\tilde{c}_1 - c_1||_\infty \leq \epsilon\\
    &||\tilde{c}_1 - v_j||_\infty \geq d_j^{t_j}, \quad j\in[m]
\end{align}
\end{subequations}
where $d_{j}^{t_j}$ is defined in the destructive control variant of \textsc{BVPM}. 
Notice that there are in total $m \cdot |f_{\text{uniq}}|$ open hypercubes involved. According to Lemma \ref{non_convex}, the representative solution set $P$ for all scenarios has at most $2m \cdot |f_{\text{uniq}}|$ choices for each dimension. The time complexity of the algorithm is exponential only in $d$.
\end{proof}

\subsection{Constant Number of Distinct Voters}

Next, we turn to the case when the number of distinct voter opinion vectors $v_j$ is bounded by a constant.
As in Section~\ref{S:bvpm}, we simplify the discussion by assuming that the number of voters is constant; generalization to an arbitrary number of voters whose opinions can be grouped into a small set $Q$ of possibilities is straightforward using the same idea as for \textsc{BVPM}.
For $l_2$ norm, we use the ball intersection algorithm as in Section~\ref{S:rvpm_const_issues}.
\begin{theorem}
When the number of voters is constant, the constructive and destructive control variants of \textsc{RVPM} can be solved in polynomial time under the $l_2$ norm for arbitrary scoring rules.
\end{theorem}
\begin{proof}
For constructive control, we solve the problem for each scenario of $(t_1,\ldots,t_m)$ (defined as in Theorem~\ref{T:rvpm_c_ci}) through finding a representative point within $\{B_c\} \cup (\bigcup_{j=1}^m \{B_j^{t_j}\})$, which takes $O(d^3)$ time using a variation of the ball intersection algorithm by \citeauthor{Crama95}~\shortcite{Crama95} as in Theorem~\ref{T:rvpm_c_ci}. 
Since there are in total $O({|f_{\text{uniq}}|}^m)$ scenarios, the problem can be solved in polynomial time. For destructive control, a similar argument holds and the problem can be solved by using a variation of the ball intersection algorithm as in Theorem \ref{T:rvpm_d_ci}.
\end{proof}
For $l_\infty$ norm, the constructive case can be solved by the application of linear programming.
\begin{theorem}
 When the number of voters is constant, the constructive control variant of \textsc{RVPM} can be solved in polynomial time under the $l_\infty$ norm for arbitrary scoring rules.
\end{theorem}
\begin{proof}
For each scenario of $(t_1,\ldots,t_m)$ (defined as in Theorem~\ref{T:rvpm_c_ci}), we solve the below linear programming:
\begin{subequations}
\begin{align}
  & ||\tilde{c}_1-{c}_1||_\infty \leq \epsilon\\
&    ||\tilde{c}_1-{v}_j||_\infty \leq d_j^{t_j},\quad j \in [m]
\end{align}
\end{subequations}
Since each LP has $O(d)$ linear constraints, and there are $O({|f_{\text{uniq}}|}^m)$ scenarios, the problem is polynomial time solvable.
Alternatively, we can check the interval endpoints similar to Lemma \ref{non_convex} for each dimension. Since $m$ is a constant, the algorithm returns $\tilde{c}_1$ in $O(d)$ time if a solution to the LP exists, or NO if not feasible. %
\end{proof}

The destructive case is somewhat more involved.
We begin with a lemma that again shows that a non-convex feasibility problem we need to solve is tractable.
\begin{lemma}
The feasibility problem in Lemma \ref{non_convex} can be solved in linear time if the number of constraints $k$ is constant.
\label{non_convex_poly}
\end{lemma}
\begin{proof}[Proof sketch]
According to Lemma \ref{non_convex}, for each dimension $j\in [d]$, we could compute set ${P}_j$ and its corresponding set ${\mathcal{S}_j = \{S_j(p_j)\mid p_j \in P_j\}}$. %
Finding a solution to the feasibility problem is equivalent to finding $S_j(p_j) \in \mathcal{S}_j$ for all $j\in[l]$, with some $l \leq d$ that satisfy $\cup_{j=1}^l S_j(p_j)= [k]$, and $[p_1,p_2,\ldots,p_l,y_{l+1},y_{l+2},\ldots,y_d]$ is a solution. Since $k$ is constant, the algorithm is linear and of complexity $O(d)$. We can also determine in linear time if none of the representative points satisfy the feasibility condition, and return NO in that case.
The detailed algorithm and proof
are provided in Supplement \ref{algo_1_des}.
\end{proof}
\begin{theorem}
When the number of voters is constant, the destructive control variant of \textsc{RVPM} can be solved in polynomial time under $l_\infty$ norm for arbitrary scoring rules.
\end{theorem}
\begin{proof}
Using Lemma \ref{non_convex_poly}, we can solve the linear feasibility problem in Equation~\eqref{E:lfp} where $m$ is constant for each scenario of $(t_1,\ldots,t_m)$ (defined as in Theorem~\ref{T:rvpm_d_ci}).
Each feasibility problem can be solved in $O(d)$ time and there are  $O({|f_{\text{uniq}}|}^m)$ scenarios. The full time complexity of the algorithm is ${O(n(d + \log(n)) + |f_{\text{uniq}}|^m (d+n))}$, or $O(n d)$ if $|f_{\text{uniq}}|$ is constant.
\end{proof}

We can extend the results to an arbitrary number of voters with a constant number of distinct opinions by the same argument as for \textsc{BVPM}.

\section{Conclusion}

We model the impact of political misinformation on elections as election control in the spatial model of voting in which an adversary manipulates perceptions of the positions of a target candidate by the voters.
Our central observation, which obtains both when issues are real-valued and binary, and for different ways we can measure distance in generating preferences over candidates based on their relative positions to voters, is that what matters is the extent of opinion diversity in the voting population.
Specifically, when voter positions on issues are highly diverse, the manipulation problem is intractable in most settings.
In contrast, when voter views can be reduced by a small number of opinion groups, the control problem becomes linear in dimension when issues are binary, and polynomial with real-valued issues.
Our characterization of the complexity landscape leaves several open questions, such as hardness of constructive control with two candidates in the setting with real-valued issues (we only show that it is tractable for $l_\infty$ norm).
Furthermore, our negative results for real-valued issues do not address the case of $l_1,$ while our positive results only apply to $l_2$ and $l_\infty$ in this setting.

Our model has several important limitations that suggest further useful future directions.
First, we assume that the same norm is used both by voters to rank candidates, and to limit the extent of perception manipulation; however, these distances may often be useful to measure in different ways.
Second, we assume that perception manipulation has identical impact on all voters.
A more sophisticated model would blend this with the election control approaches in which misinformation spreads through a social network, with only a subset of voters impacted, potentially in different ways.

\paragraph{Acknowledgements}
This work was partially supported by the National Science Foundation (grants IIS-1905558, IIS-1903207, and IIS-1939677) and Amazon.
\bibliographystyle{named}
\bibliography{ecpm}

\begin{thebibliography}{}

\bibitem[\protect\citeauthoryear{Anshelevich and Postl}{2016}]{Anshelevich16}
Elliot Anshelevich and John Postl.
\newblock Randomized social choice functions under metric preferences.
\newblock In {\em International Joint Conference on Artificial Intelligence},
  page 46–52, 2016.

\bibitem[\protect\citeauthoryear{Anshelevich \bgroup \em et al.\egroup
  }{2018}]{Anshelevich18}
Elliot Anshelevich, Onkar Bhardwaj, Edith Elkind, John Postl, and Piotr
  Skowron.
\newblock Approximating optimal social choice under metric preferences.
\newblock {\em Artificial Intelligence}, 264:27--51, 2018.

\bibitem[\protect\citeauthoryear{Bartholdi \bgroup \em et al.\egroup
  }{1992}]{Bartholdi92}
John~J. Bartholdi, Craig~A. Tovey, and Michael~A. Trick.
\newblock How hard is it to control an election?
\newblock {\em Mathematical and Computer Modelling}, 16(8):27--40, 1992.

\bibitem[\protect\citeauthoryear{Binkele-Raible \bgroup \em et al.\egroup
  }{2014}]{Binkele2014}
Daniel Binkele-Raible, G{\'a}bor Erd{\'e}lyi, Henning Fernau, Judy Goldsmith,
  Nicholas Mattei, and J{\"o}rg Rothe.
\newblock The complexity of probabilistic lobbying.
\newblock {\em Discrete Optimization}, 11:1--21, 2014.

\bibitem[\protect\citeauthoryear{Bredereck \bgroup \em et al.\egroup
  }{2016}]{Bredereck16}
Robert Bredereck, Piotr Faliszewski, Rolf Niedermeier, and Nimrod Talmon.
\newblock Large-scale election campaigns: Combinatorial shift bribery.
\newblock {\em Journal of Artificial Intelligence Research}, 55:603--652, 2016.

\bibitem[\protect\citeauthoryear{Caldwell \bgroup \em et al.\egroup
  }{2019}]{Caldwell16}
L.~A. Caldwell, H.~Przybyla, and K.~Stewart.
\newblock Senate intelligence report finds 'extensive' {R}ussian election
  interference.
\newblock In {\em {NBC News}}. 2019.

\bibitem[\protect\citeauthoryear{Castiglioni \bgroup \em et al.\egroup
  }{2020}]{Castiglioni20}
Matteo Castiglioni, Diodato Ferraioli, and Nicola Gatti.
\newblock Election control in social networks via edge addition or removal.
\newblock In {\em AAAI Conference on Artificial Intelligence}, pages
  1878--1885, 2020.

\bibitem[\protect\citeauthoryear{Chen \bgroup \em et al.\egroup
  }{2017}]{Chen17}
Jiehua Chen, Piotr Faliszewski, Rolf Niedermeier, and Nimrod Talmon.
\newblock Elections with few voters: Candidate control can be easy.
\newblock {\em Journal of Artificial Intelligence Research}, 60(1):937–1002,
  2017.

\bibitem[\protect\citeauthoryear{Christian \bgroup \em et al.\egroup
  }{2007}]{Christian2007}
Robin Christian, Mike Fellows, Frances Rosamond, and Arkadii Slinko.
\newblock On complexity of lobbying in multiple referenda.
\newblock {\em Review of Economic Design}, 11(3):217--224, 2007.

\bibitem[\protect\citeauthoryear{Crama and Ibaraki}{1997}]{Crama97}
Yves Crama and Toshihide Ibaraki.
\newblock Hitting or avoiding balls in euclidean space.
\newblock {\em Annals of Operations Research}, 69:47--64, 1997.

\bibitem[\protect\citeauthoryear{Crama \bgroup \em et al.\egroup
  }{1995}]{Crama95}
Yves Crama, Pierre Hansen, and Brigitte Jaumard.
\newblock Complexity of product positioning and ball intersection problems.
\newblock {\em Mathematics of Operations Research}, 20(4):885--894, 1995.

\bibitem[\protect\citeauthoryear{Davis and Hinich}{1968}]{Davis68}
Otto~A. Davis and Melvin~J. Hinich.
\newblock On the power and importance of the mean preference in a mathematical
  model of democratic choice.
\newblock {\em Public Choice}, 5(1):59--72, 1968.

\bibitem[\protect\citeauthoryear{Enelow and Hinich}{1984}]{Enelow84}
James~M. Enelow and Melvin~J. Hinich.
\newblock {\em The Spatial Theory of Voting: An Introduction}.
\newblock Cambridge University Press, 1984.

\bibitem[\protect\citeauthoryear{Erdélyi \bgroup \em et al.\egroup
  }{2015}]{Erdelyi15}
Gábor Erdélyi, Michael~R. Fellows, Jörg Rothe, and Lena Schend.
\newblock Control complexity in bucklin and fallback voting: A theoretical
  analysis.
\newblock {\em Journal of Computer and System Sciences}, 81(4):632--660, 2015.

\bibitem[\protect\citeauthoryear{Estornell \bgroup \em et al.\egroup
  }{2020}]{Estornell20}
Andrew Estornell, Sanmay Das, Edith Elkind, and Yevgeniy Vorobeychik.
\newblock Election control by manipulating issue significance.
\newblock In {\em Conference on Uncertainty in Artificial Intelligence}, pages
  340--349, 2020.

\bibitem[\protect\citeauthoryear{Faliszewski and Rothe}{2016}]{Faliszewski16}
Piotr Faliszewski and J{\"{o}}rg Rothe.
\newblock Control and bribery in voting.
\newblock In Felix Brandt, Vincent Conitzer, Ulle Endriss, J{\'{e}}r{\^{o}}me
  Lang, and Ariel~D. Procaccia, editors, {\em Handbook of Computational Social
  Choice}, pages 146--168. Cambridge University Press, 2016.

\bibitem[\protect\citeauthoryear{Harper \bgroup \em et al.\egroup
  }{2019}]{Harper19}
T.~Harper, C.~Wheeler, and R.~Kerbaj.
\newblock Revealed: the {R}ussia report.
\newblock In {\em The Sunday Times}. 2019.

\bibitem[\protect\citeauthoryear{Hemaspaandra \bgroup \em et al.\egroup
  }{2007}]{Hemaspaandra07}
Edith Hemaspaandra, Lane~A. Hemaspaandra, and Jörg Rothe.
\newblock Anyone but him: The complexity of precluding an alternative.
\newblock {\em Artificial Intelligence}, 171:255--285, 2007.

\bibitem[\protect\citeauthoryear{Imai and Asano}{1983}]{Imai83}
Hiroshi Imai and Takao Asano.
\newblock Finding the connected components and a maximum clique of an
  intersection graph of rectangles in the plane.
\newblock {\em Journal of Algorithms}, 4(4):310--323, 1983.

\bibitem[\protect\citeauthoryear{Kempe \bgroup \em et al.\egroup
  }{2003}]{Kempe03}
David Kempe, Jon Kleinberg, and {\'E}va Tardos.
\newblock Maximizing the spread of influence through a social network.
\newblock In {\em ACM SIGKDD International Conference on Knowledge Discovery
  and Data Mining}, pages 137--146, 2003.

\bibitem[\protect\citeauthoryear{Khetani-Shah and
  Deutsch}{2019}]{Khetani-Shah19}
S.~Khetani-Shah and J.~Deutsch.
\newblock Brexit timeline: From referendum to {EU} exit.
\newblock In {\em Politico Pro}. 2019.

\bibitem[\protect\citeauthoryear{Lee}{1983}]{Lee83}
Der-Tsai Lee.
\newblock Maximum clique problem of rectangle graphs.
\newblock {\em Advances in Computing Research}, 1:91--107, 1983.

\bibitem[\protect\citeauthoryear{Lokshtanov}{2009}]{Lokshtanov09}
Daniel Lokshtanov.
\newblock {\em New methods in parameterized algorithms and complexity}.
\newblock PhD thesis, University of Bergen, Norway, 2009.

\bibitem[\protect\citeauthoryear{Lu \bgroup \em et al.\egroup }{2019}]{Lu19}
Jasper Lu, David~Kai Zhang, Zinovi Rabinovich, Svetlana Obraztsova, and
  Yevgeniy Vorobeychik.
\newblock Manipulating elections by selecting issues.
\newblock In {\em International Conference on Autonomous Agents and Multiagent
  Systems}, page 529–537, 2019.

\bibitem[\protect\citeauthoryear{Megiddo}{1990}]{Megiddo90}
Nimrod Megiddo.
\newblock On the complexity of some geometric problems in unbounded dimension.
\newblock {\em Journal of Symbolic Computation}, 10(3-4):327--334, 1990.

\bibitem[\protect\citeauthoryear{Menton}{2012}]{Menton12}
Curtis Menton.
\newblock Normalized range voting broadly resists control.
\newblock {\em Theory of Computing Systems}, 53(4):507--531, 2012.

\bibitem[\protect\citeauthoryear{Robb}{2017}]{Robb17}
Amanda Robb.
\newblock Anatomy of a fake news scandal.
\newblock In {\em Rolling Stone}. 2017.

\bibitem[\protect\citeauthoryear{Wilder and Vorobeychik}{2018}]{Wilder18}
Bryan Wilder and Yevgeniy Vorobeychik.
\newblock Controlling elections through social influence.
\newblock In {\em International Conference on Autonomous Agents and Multiagent
  Systems}, page 265–273, 2018.

\end{thebibliography}

\newpage
\appendix
\section*{Supplementary Materials}

\section{Missing Proofs}

\subsection{Proof of Lemma~\ref{constructiverank}}

For each voter $v_j, j \in [m]$, we rank candidates $c_i ~({i\geq 2})$ by their distances to $v_j$ from closest to furthest. %
Without loss of generality, we assume candidates $c_2,c_3,\ldots,c_n$ are ranked from closest to furthest w.r.t. $v_j$. Then when $c_1$ is inserted into this sequence and is ranked $r_j$, the score each candidate receives from $v_j$ is as below:

\begin{table}[H]
\centering
\begin{tabular}{|c|c|c|c|c|c|c|}
\hline
$c_{2}$ & ... & $c_{r_j}$   & $c_1$ & $c_{r_j+1}$     & ... & $c_{n}$ \\
\hline
$f(1)$ & ... & $f(r_j-1)$ & $f(r_j)$ & $f(r_j+1)$ &  ...  & $f(n)$ \\
\hline
\end{tabular}
\end{table}

If we move $c_1$ from ranking position $r_j$ to $r'_j$, ${r'_j \leq r_j}$, since $f$ is a non-increasing function, the score $c_1$ receives from $v_j$ will increase from $f(r_j)$ to $f(r'_j)$; the scores candidates $c_{r_j+1},c_{r_j+2},\ldots,c_{n}$ and $c_{2},c_{3},\ldots,c_{r'_j}$ receive from $v_j$ will not change; for candidates $c_{r'_j+1}, c_{r'_j+2},\ldots,c_{r_j}$, since their rankings will decrease by $1$, the scores they receive from $v_j$ will not increase. This means if we move $c_1$ from ranking position $(r_1,r_2,\ldots,r_m)$ to $(r'_1,r'_2,\ldots,r'_m)$, ${r'_j \leq r_j}, \forall j \in [m]$, the total score $c_1$ receives will not decrease, while the total scores other candidates receive will not increase. Since $c_1$ wins the election with ranking position $(r_1,r_2,\ldots,r_m)$, $c_1$ will also win the election with ranking position $(r'_1,r'_2,\ldots,r'_m)$.

\subsection{Proof of Lemma~\ref{constructiverep}}

Without loss of generality, we assume candidates $c_2,c_3,\ldots,c_n$ are ranked from closest to furthest w.r.t. voter $v_j$. We next insert $c_1$ into this sequence and assume $c_1$ receives a score of $f(s_{i+1})$ from $v_j$, where $i\in\{0,1,\ldots,r-1\}$. After the insertion, the score each candidate receives from $v_j$ is as below:
\begin{table}[H]
\centering
\begin{tabular}{c||c||c||c}
\hline
$...,c_{s_i+1}$ & $... ,c_1, ... $ & $ c_{s_{i+1}+1}, ..., c_{s_{i+2}}$   & $c_{s_{i+2}+1},...$      \\
\hline
$f(s_i)$ & $f(s_{i+1})$ & $f(s_{i+2})$ & $f(s_{i+3})$  \\
\hline
\end{tabular}
\end{table}

Notice that whichever ranking position $c_1$ takes among ${s_i +1},s_i +2,\ldots,s_{i+1}$, the final score each candidate receives from $v_j$ are the same, meaning these ranking positions are equivalent. 
Since the solution set of getting a ranking position higher than $s_{i+1}$ covers the solution set of getting a ranking position higher than $s_i +1,s_i +2,\ldots,s_{i+1}-1$, we only need to consider $s_{i+1}$.
\subsection{Proof of Theorem \ref{rvpm_des_lp} ($l_p$ norm, $1<p<\infty$)}
\label{S:rvpm_constructive_hardness_lp}
Checking whether target candidate $\tilde{c}_1$ loses the election is clearly in P. We next prove hardness by reduction from \textsc{3-SAT}. The proof is in $\mathbb{R}^{d+1}$ space.

\begin{definition}[\textsc{3-SAT}]
Given a set of variables $X$, a set of clauses $C$ over $X$ where $\forall c \in C$ has ${|c|=3}$, is there an assignment to the variables $X$ that satisfies all the clauses?
\end{definition}

Let target candidate $c_1=[0,0,\ldots,0,0]$, rival candidate $c_2=[0,0,\ldots,0,a]$ with ${a=((\frac{1}{d}+1)^p} + (\frac{1}{d})^p(d-1)-1)^\frac{1}{p}$, and $\epsilon = d^{\frac{1}{p}-1}$. Given a \textsc{3-SAT} instance with $d-1$ variables $\{X_1,\ldots,X_{d-1}\}$ and $r$ clauses $\{E_1, E_2,\ldots,E_r\}$, we create a voter profile $v_D = [0,0,\ldots,0,-a]$ which is adopted by $r$ voters. Note that since the distance between $c_2$ and $v_D$ is $2a$, while the maximum distance between $\tilde{c}_1$ and $v_D$ is ${a + d^{\frac{1}{p}-1} < 2a}$ ($\max~||\tilde{c}_1-v_D||_p$, s.t., $||\tilde{c}_1||_p\leq \epsilon$), wherever $\tilde{c}_1$ moves, $v_D$ will always vote for $c_1$. Let $e_i\in \mathbb{R}^{d+1}$ be the unit vector that has $1$ in the $i$-th position. We create $2d$ voters $\{v_{e_i}, v_{-e_i}\}_{i=1}^d$ where $v_{e_i} = e_i, v_{-e_i}=-{e_i}, i \in [d]$. Next we create $r$ voters $\{{v_i}\}_{i=1}^{r}$ where each $v_i$ corresponds to clause $E_i, i \in [r]$. Given fixed parameters $\alpha>0, l>0$ (which we will specify later), for each $v_{i}$, $v_{i,j} = -\alpha$ if $X_j$ is included in clause $E_i$, $v_{i,j} = \alpha$ if $\overline{X}_j$ is included in clause $E_i, j\in[d-1]$; $v_{i,d}=l \alpha$; $0$ otherwise. 

We now specify parameters $\alpha$ and $l$. We chose $\alpha>0$ and $l>0$ so that they satisfy the below inequality:
\begin{equation}
    \begin{split}
    &(\frac{1}{d}+\alpha)^p + 2|\frac{1}{d}-\alpha|^p +(d-4)(\frac{1}{d})^p + (\frac{1}{d}+l\alpha)^p\\
    &> a^p +3\alpha^p + (l\alpha)^p\\
    &> 3|\frac{1}{d}-\alpha|^p +(d-4)(\frac{1}{d})^p + (\frac{1}{d}+l\alpha)^p.
    \end{split}
    \label{des_2_lp_param}
\end{equation}
We claim such parameters must exist. We first move term $3\alpha^p + (l\alpha)^p$ to the right and left side of the inequality. Since $d$ and $p$ are given, $a^p$ is a fixed value in $\mathbb{R}$, and the right and left side of the inequality are continuous w.r.t. $\alpha$ and $l$, with gap equals to $(\frac{1}{d}+\alpha)^p - |\frac{1}{d}-\alpha|^p >0$. This means if we could find (a) $l>0$, $\alpha_1>0$ that satisfy ${(\frac{1}{d}+\alpha_1)^p} + {2|\frac{1}{d}-\alpha_1|^p} +{(d-4)(\frac{1}{d})^p} + {(\frac{1}{d}+l\alpha_1)^p }-3\alpha_1^p -(l\alpha_1)^p > a^p$ and (b) $\alpha_2>0$ that satisfy ${(\frac{1}{d}+\alpha_2)^p} + {2|\frac{1}{d}-\alpha_2|^p} +{(d-4)(\frac{1}{d})^p} + {(\frac{1}{d}+l\alpha_2)^p }-3\alpha_2^p -(l\alpha_2)^p < a^p$, by intermediate value theorem there must exist $\alpha>0$ that satisfies inequality (\ref{des_2_lp_param}). Notice that regardless of the value of $l$, we could always find $\alpha_1>0$ large enough so that the condition (a) is satisfied; as for condition (b), by %
having $\alpha_2 = \frac{1}{d}$, we could find $l>0$ that satisfies the condition.

Now we have $2r+2d$ voters in total, and they all vote for $c_1$. RVPM moves $c_1$ to $\tilde{c}_1$ so that $\tilde{c}_1$ loses at least $r+d$ voters and loses the election. Recall that $v_D$ will always vote for $\tilde{c}_1$. We next find a solution $\tilde{c}_1$ that will lose $d$ votes in $\{v_{e_i}, v_{-e_i}\}_{i=1}^d$ and all $r$ voters in $\{{v_i}\}_{i=1}^{r}$. 

Notice that the distance between $c_2$ and any voter in $\{v_{e_i}, v_{-e_i}\}_{i=1}^d$ is $(a^p+1)^\frac{1}{p}$. We claim that for $\tilde{c}_1$ to lose $d$ voters in $\{v_{e_i}, v_{-e_i}\}_{i=1}^d$ (which is maximum), $\tilde{c}_1$ has to have format $[\pm \frac{1}{d},\pm \frac{1}{d},\ldots, \pm \frac{1}{d},0]$. Given the budget constraint ${||\tilde{c}_1||_p\leq \epsilon}$, the solution of $\tilde{c}_1$ losing maximum number of voters can always be achieved at the constraint boundary, since we can always move $\tilde{c}_{1,d+1}$ away from $0$ until ${||\tilde{c}_1||_p = \epsilon}$ and this will only increase the distance between $\tilde{c}_1$ and any voter in $\{v_{e_i}, v_{-e_i}\}_{i=1}^d\cup\{{v_i}\}_{i=1}^{r}$. 
With ${||\tilde{c}_1||_p = \epsilon}$, we have
    ${||\tilde{c}_1-v_{e_i}||_p^p}={|\tilde{c}_{1,i}-1|^p+d^{1-p}-|\tilde{c}_{1,i}|^p}$,
   ${||\tilde{c}_1-v_{-e_i}||_p^p}={|\tilde{c}_{1,i}+1|^p+d^{1-p}-|\tilde{c}_{1,i}|^p}$.
Recall that in destructive control $\tilde{c}_1$ will lose the voter in case of tie-breaking. Since for $\tilde{c}_1$ to lose $v_{e_i}$, due to the budget constraint, we must have $\tilde{c}_{1,i}\leq -\frac{1}{d}$; and for $\tilde{c}_1$ to lose $v_{-e_i}$ we must have $\tilde{c}_{1,i}\geq \frac{1}{d}$; $\tilde{c}_1$ can only lose maximum one voter between $v_{e_i}$ and $v_{-e_i}, {i \in [d]}$. Given the budget constraint, in order for $\tilde{c}_1$ to lose $d$ voters in $\{v_{e_i}, v_{-e_i}\}_{i=1}^d$, $\tilde{c}_1$ has to have the above format.

Next we show that finding a $\tilde{c}_1$ to lose all $r$ voters in $\{{v_i}\}_{i=1}^{r}$ is equivalent to finding a solution to \textsc{3-SAT}. 

Given a \textsc{3-SAT} solution $\{X_1,\ldots,X_{d-1}\}$, we let $\tilde{c}_{1,i} = \frac{1}{d}$ if $X_i$ is TRUE, $\tilde{c}_{1,i} = -\frac{1}{d}$ if $X_i$ is FALSE, ${i\in[d-1]}$; $\tilde{c}_{1,d}=-\frac{1}{d}$ and $\tilde{c}_{1,d+1}=0$. It is easy to check $\tilde{c}_1$ satisfies the budget constraint. For any clause $E_i$, $i \in [r]$ that has $t$ literals to be TRUE $(1\leq t\leq 3)$, from inequality (\ref{des_2_lp_param}) we have ${||\tilde{c}_1-v_i||_p } \geq  ||c_2-v_i||_p$, meaning $\tilde{c}_1$ loses voter $v_{i}$ and is a solution to RVPM.

Given a RVPM solution $\tilde{c}_1$, since $\tilde{c}_1$ has format $[\pm \frac{1}{d},\pm \frac{1}{d},\ldots, \pm \frac{1}{d},\pm \frac{1}{d},0]$, we can construct a \textsc{3-SAT} solution $\{X_1,\ldots,X_{d-1}\}$ where variables $X_i =\text{TRUE}$ if $\tilde{c}_{1,i} = \frac{1}{d}$, or $X_i =\text{FALSE}$ if $\tilde{c}_{1,i} = -\frac{1}{d}$, ${i\in[d-1]}$. Since $v_{i,d}=l \alpha$, we assume $\tilde{c}_{1,d}=-\frac{1}{d}$. For any clause $E_i, {i \in[r]}$ that has $t$ literals to be TRUE $(0\leq t\leq 3)$, since we need $\tilde{c}_1$ to lose voter $v_{i}$, meaning ${||\tilde{c}_1-v_i||_p \geq ||c_2-v_i||_p}$, from inequality (\ref{des_2_lp_param}) we must have $t\geq1$.
\subsection{Proof of Theorem \ref{rvpm_des_lp} ($l_\infty$ norm)}
\label{S:rvpm_constructive_hardness_linf}
For RVPM $l_\infty$ norm,
again, it is easy to see that it's in NP. We now show hardness again by reduction from \textsc{3-SAT}. The proof is in $\mathbb{R}^{d+1}$ space.

We set target candidate $c_1 = [0,0,\ldots,0,0]$,  rival candidate $c_2 = [0,0,\ldots,0,2]$ and let $\epsilon=1$. We are given a \textsc{3-SAT} instance with $d-1$ variables $\{X_1,X_2,\ldots,X_{d-1}\}$ and $r$ clauses $\{E_1, E_2,\ldots,E_r\}$. We create $r$ voters $\{{v_i}\}_{i=1}^{r}$ where each $v_i$ corresponds to clause $E_i, i \in[r]$. For each $v_{i}$, we set $v_{i,j} = -1$ if $X_j$ is included in clause $E_i$, or $v_{i,j} = 1$ if $\overline{X}_j$ is included in clause $E_i$, $j \in[d-1]$; $0$ otherwise. We also create $r$ dummy voters at $[0,0,\ldots,0,0]$; those dummy voters will always vote for $\tilde{c}_1$ since their distances to $c_2$ are always 2 while their distances to $\tilde{c}_1$ are within $1$ wherever $\tilde{c}_1$ moves. We now have $2r$ voters, and they all vote for $c_1$. \textsc{RVPM} will change $c_1$ to $\tilde{c}_1$ so that $\tilde{c}_1$ loses at least $r$ voters.

Given a \textsc{3-SAT} solution $\{X_1,\ldots,X_{d-1}\}$, we let $\tilde{c}_{1,i} = 1$ if $X_i $ is TRUE, and $\tilde{c}_{1,i} = -1$ if $X_i$ is  FALSE, $i\in[d-1]$; $\tilde{c}_{1,d}=0$. It is easy to check $||\tilde{c}_1-v_i||_\infty=||{c}_2-v_i||_\infty = 2, \forall i\in[r]$. $\tilde{c}_1$ loses $r$ voters and loses the election.

Given a RVPM solution $\tilde{c}_1$, since $\tilde{c}_1$ loses the election, we must have $||\tilde{c}_1-v_i||_\infty\geq ||{c}_2-v_i||_\infty = 2, \forall i\in[r]$. Since $v_i\in \{0,\pm 1\}^d, \forall i \in [r]$ and $||\tilde{c}_1||_\infty\leq 1$, $\tilde{c}_1$ can only lose $r$ voters by having
$||\tilde{c}_1-v_i||_\infty= 2, \forall i\in[r]$. Since $v_{i,d}=0$, we must have $\tilde{c}_{1,i}\in \{-1, 1\}, i \in [d-1]$. By setting $X_i =\text{TRUE}$ if $\tilde{c}_{1,i}=1$, $X_i =\text{FALSE}$ if $\tilde{c}_{1,i}=-1, i \in [d-1]$, we have a solution $\{X_1,\ldots,X_{d-1}\}$ to \textsc{3-SAT}.
\subsection{Proof of Theorem \ref{T:cc_real_manyc_lp} ($l_p$ norm, $1<p<\infty$)}
\label{S:supp_cc_manyc_lp}

Checking whether a given target candidate $\tilde{c}_1$ wins the election is clearly in P.
We now show hardness by reduction from \textsc{3-SAT}. The proof is in $\mathbb{R}^d$ space.%

Let $e_i\in \mathbb{R}^{d}$ be the unit vector that has $1$ in the $i$-th position, $d'=d-1$. We first introduce the below lemma.
\begin{lemma}
The smallest ball that encloses $\{e_i\}_{i=1}^{d'}$ is centered at $c\sum_{i=1}^{d'}e_i$, where $c = \frac{1}{1+(d'-1)^{1/(p-1)}}$, and the radius of the ball is $r_{d'} = ((d'-1)c^p+(1-c)^p)^{1/p}$.\label{lemma:cl}
\end{lemma}
\begin{proof}
By symmetry, the center of the smallest ball has the form $c\sum_{i=1}^{d'}e_i~(0\leq c \leq 1)$, and its distance to $e_i, i \in[d']$ is 
\[||e_i- c\sum_{i=1}^{d'} e_i||_p = ((d'-1)c^p+(1-c)^p)^{1/p}.\]

Since the second order derivative of the distance w.r.t. $c$ is always greater than $0$, by having the first order derivative equals zero, we have $c = \frac{1}{1+(d'-1)^{1/(p-1)}}$.
\end{proof}

Notice that we could always change any $e_i$ to $-e_i, i \in [d']$ in Lemma \ref{lemma:cl} the conclusion still holds.

Next we construct the voter set and candidate set for RVPM problem.
Let target candidate $c_T = [0,0,\ldots,0,0]$ and $\epsilon=c\cdot (d')^\frac{1}{p}$ ($c$ is defined in Lemma \ref{lemma:cl}).
Given a \textsc{3-SAT} instance with $d-2$ variables $\{X_1,\ldots,X_{d-2}\}$ and $r$ clauses $\{E_1, E_2,\ldots,E_r\}$, we create $r$ voters $\{v_i\}_{i=1}^r$ and $r$ candidates $\{c_i\}_{i=1}^r$ where each $v_i$ and $c_i$ corresponds to clause $E_i, i \in [r]$. Given fixed parameters $l>0, \alpha>0$ (which we will specify later), if $X_j$ is included in clause $E_i$, we set $v_{i,j}=c_{i,j}=\alpha$, and if $\overline{X}_j$ is included in clause $E_i$, we set $v_{i,j}=c_{i,j}=-\alpha, j \in [d-2]$; $v_{i,d'}=c_{i,d'}=l \alpha$; $c_{i,d}=r_{d'}$ ($r_{d'}$ is defined in Lemma \ref{lemma:cl}); $0$ otherwise. We next create $2d'$ voters $\{v_{e_i},v_{-e_i}\}_{i=1}^{d'}$ and $2d'$ candidates $\{c_{e_i},c_{-e_i}\}_{i=1}^{d'}$ where $\{v_{e_i}, c_{e_i}\}$ corresponds to $e_i$ and $\{v_{-e_i}, c_{-e_i}\}$ corresponds to $-e_i$. For any $i\in[d']$, we set $v_{{e_i},i}= c_{{e_i},i}=1$ and $v_{{-e_i},i}= c_{{-e_i},i}=-1$; $c_{{e_i},d}=c_{{-e_i},d}=r_{d'}$; $0$ otherwise. We also create a rival candidate $c_R = [M,0,0,\ldots,0,0]$ and $d'+r$ dummy voters also at $[M,0,0,\ldots,0,0]$, where $M$ is large enough so that those dummy voters will always vote for $c_R$. Notice that as of current, any candidates created from \textsc{3-SAT} clauses and from $\{\pm e_i\}_{i=1}^{d'}$ receives one vote, that is, $v_j$ votes for $c_j, j\in[r]$; $v_{e_i}$ votes for $c_{e_i}$ and $v_{-e_i}$ votes for $c_{-e_i}, i \in [d']$; the distance between each pair of voter and candidate is $r_{d'}$. RVPM moves $c_1$ to $\tilde{c}_1$ so that $\tilde{c}_1$ wins at least $r+d'$ voters and wins the election.

Recall that when we create the voters and candidates w.r.t. \textsc{3-SAT} clauses, we have two fixed parameters $l>0, \alpha>0$. We chose those two parameters to satisfy the below inequality:
\begin{equation}
\begin{split}
        |\alpha-c|^p+2(\alpha+c)^p+(d'-4)c^{p} +|l\alpha-c|^p < &\\r_{d'}^p< 3(\alpha+c)^p+(d'-4)c^{p} +|l\alpha-c|^p.&\label{rvpm_2_ball}
\end{split}
\end{equation}
We claim that parameters $l>0, \alpha>0$ that satisfy inequality (\ref{rvpm_2_ball}) must exist. Since $d'$ and $p$ are given, $r_{d'}$ is a fixed value in $\mathbb{R}$, and the right and left side of the inequality are continuous w.r.t. $\alpha$ and $l$, with gap equals to $(\alpha+c)^p - |\alpha-c|^p >0$. This means if we could find (a) $l>0$, $\alpha_1>0$ that satisfy ${|\alpha_1-c|^p}+{2(\alpha_1+c)^p}+{(d'-4)c^{p}} +{|l\alpha_1-c|^p} > r_{d'}^p$ and (b) $\alpha_2>0$ that satisfy ${|\alpha_2-c|^p}+{2(\alpha_2+c)^p+(d'-4)c^{p}} +{|l\alpha_2-c|^p} < r_{d'}^p$, by intermediate value theorem there must exist $\alpha>0$ that satisfies inequality (\ref{rvpm_2_ball}). Notice that regardless of the value of $l$, we could always find $\alpha_1>0$ large enough so that condition (a) is satisfied. Next we find $l>0$ and $\alpha_2>0$ that satisfy condition (b). Plug in $r_{d'}^p = (d'-1)c^p+(1-c)^p$, we have $|l\alpha_2-c|^p <3c^{p} + (1-c)^p-|\alpha_2-c|^p - 2|\alpha_2+c|^p$. Since $0<c<\frac{1}{2}$, there exists $\alpha_2$ such that $3c^{p} + (1-c)^p-|\alpha_2-c|^p - 2|\alpha_2+c|^p>0$, we can then find $l>0$ such that the inequality holds (e.g., $l=c/\alpha_2)$.

With parameters $l>0$ and $\alpha>0$ chosen and fixed, next we show that finding a $\tilde{c}_1$ to win the election (i.e., win at least $r+d'$ voters) is equivalent to finding a solution to \textsc{3-SAT}. 

Given a \textsc{3-SAT} solution $\{X_1,\ldots,X_{d-2}\}$, we let ${\tilde{c}_{T,i} = c}$ if $X_i$ is TRUE, $\tilde{c}_{T,i} = -c$ if $X_i$ is FALSE, ${i\in[d-2]}$; $\tilde{c}_{T,d'} = c$; $\tilde{c}_{T,d} = 0$. It is easy to check $\tilde{c}_{T}$ satisfies the budget constraint and win $d'$ voters in $\{v_{e_i},v_{-e_i}\}_{i=1}^{d'}$ by having a distance of $r_{d'}$. For any clause $E_i, i \in [r]$ that has $t$ literals to be TRUE $(1\leq t \leq 3)$, by inequality (\ref{rvpm_2_ball}) we have 
\begin{subequations}
\begin{align*}
&||v_i - \tilde{c}_T||_p\\
=&t|\alpha-c|^p+(3-t)(\alpha+c)^p+(d'-4)c^{p} +|l\alpha-c|^p\\
<&r_{d'}.
\end{align*}
\end{subequations}
$\tilde{c}_T$ wins $r+d'$ voters and wins the election.

Given a RVPM solution $\tilde{c}_{T}$, since $\tilde{c}_{T}$ needs to win at least $r+d'$ voters, this means $\tilde{c}_{T}$ has to have a distance of no more than $r_{d'}$ to at least $d$ voters in $\{v_{e_i},v_{-e_i}\}_{i=1}^{d'}$. By Lemma \ref{lemma:cl} and budget constraint requirement, $\tilde{c}_{T}$ has to have format ${\tilde{c}_{T,i}\in \{-c,c\}}, {i \in [d-1]}$, $\tilde{c}_{T,d}=0$, and $\tilde{c}_{T}$ wins only $d'$ voters in $\{v_{e_i},v_{-e_i}\}_{i=1}^{d'}$. Since $v_{i,d'}=l\alpha, i \in[r]$, we assume $\tilde{c}_{T,d'}=c$. We can then construct a \textsc{3-SAT} solution $\{X_1,\ldots,X_{d-2}\}$ where variables $X_i =\text{TRUE}$ if $\tilde{c}_{1,i} = c$, $X_i =\text{FALSE}$ if $\tilde{c}_{1,i} = -c$, ${i\in[d-2]}$. For any clause $E_i, {i \in[r]}$ that has $t$ literals to be TRUE $(0\leq t\leq 3)$, the distance between $\tilde{c}_T$ and $v_i$ is
\begin{subequations}
\begin{align*}
&||v_i - \tilde{c}_T||_p\\
=&t|\alpha-c|^p+(3-t)|\alpha+c|^p+(d'-4)c^{p} +|l\alpha-c|^p.
\end{align*}
\end{subequations}
Since $\tilde{c}_{T}$ needs to win all $r$ voters in $\{v_i\}_{i=1}^r$, we must have $||v_i - \tilde{c}_T||_p\leq r_{d'}$, by inequality (\ref{rvpm_2_ball}), this implies $t\geq1$. 

This proof is based on the covering by two balls problem in \citeauthor{Megiddo90}~\shortcite{Megiddo90}. While the original problem is under $l_2$ norm, we generalize it to $l_p$ with $p \in (1,\infty)$.

\subsection{Proof of Theorem~\ref{T:cc_real_manyc_lp} ($l_\infty$ norm)}
\label{S:supp_cc_manyc}
First, it is easy to see that the problem is in NP.
We now show hardness by reduction from \textsc{3-SAT}. The proof is in $\mathbb{R}^d$ space. 

We set target candidate $c_1 = [0,0,\ldots,0,0]$ and $\epsilon=1/2$. We are given a \textsc{3-SAT} instance with $d-1$ variables $\{X_1,X_2,\ldots,X_{d-1}\}$ and $r$ clauses $\{E_1, E_2,\ldots,E_r\}$. For each clause, we create $7$ sets of voters and candidates that each corresponds to a solution to the clause to be TRUE. For example, for clause $(X_i\lor \overline{X}_j \lor X_k)$ to be TRUE, we can have 
\begin{align*}
    \begin{split}
        &X_i = 1, X_j = 1, X_k = 1\\
        &X_i = 1, X_j = 1, X_k = -1\\
        &X_i = 1, X_j = -1, X_k = 1\\
        &X_i = 1, X_j = -1, X_k = -1\\
        &X_i = -1, X_j = 1, X_k = 1\\
        &X_i = -1, X_j = -1, X_k = 1\\
        &X_i = -1, X_j = -1, X_k = -1\\
    \end{split}
\end{align*}
where $1$ represents TRUE and $-1$ represents FALSE. For each such solution, we create a voter and a candidate where $v_{\cdot,i} = c_{\cdot,i} =1$ if $X_i = 1$, or $v_{\cdot,i} = c_{\cdot,i} =-1$ if $X_i = -1, i \in [d-1]$; $c_{\cdot,d} = 1/2$; $0$ otherwise. The distance between each voter and candidate created is $1/2$. In total we create $7r$ voters and candidates in this fashion. Next we create a rival candidate $c_2 = [5,5,\ldots,5,5]$ and $r$ dummy voters at $[5,5,\ldots,5,5]$. Those dummy voters will always vote for $c_2$ who then receives $r$ votes. This means in order for $c_1$ to win the election, $c_1$ needs to win at least $r$ voters created from the \textsc{3-SAT} clauses. Notice that currently no voter votes for $c_1$. \textsc{RVPM} will change $c_1$ to $\tilde{c}_1$ so that $\tilde{c}_1$ wins at least $r$ voters.

We next show finding a solution $\tilde{c}_1$ to RVPM is equivalent to finding a solution to \textsc{3-SAT}.

Given a \textsc{3-SAT} solution $\{X_1,\ldots,X_{d-1}\}$, we let $\tilde{c}_{1,i} = 1/2$ if $X_i =\text{TRUE}$, or $\tilde{c}_{1,i} = -1/2$ if $X_i =\text{FALSE}$, ${i\in[d-1]}$; $\tilde{c}_{1,d}=0$. $\tilde{c}_1$ wins $r$ voters and wins the election.

Given a RVPM solution $\tilde{c}_1$, since $\tilde{c}_1$ wins the election, $\tilde{c}_1$ needs to win at least $r$ voters created from the \textsc{3-SAT} clauses, by having a distance of no more than $1/2$ to those $r$ voters. Notice that $\tilde{c}_1$ can only win over maximum $1$ voter from each clause. Due to the budget constraint, $\tilde{c}_1$ has to have format $\tilde{c}_{1,i}\in \{-1/2,1/2\}, i \in [d-1]$. We let $X_i =\text{TRUE}$ if $\tilde{c}_{1,i}=1/2$, $X_i =\text{FALSE}$ if $\tilde{c}_{1,i}=-1/2, i \in [d-1]$, and $\{X_1,\ldots,X_{d-1}\}$ is the solution to \textsc{3-SAT}.

\subsection{Proof of Theorem~\ref{T:rvpm_cc_2}}
\label{S:supp_rvpm_cc_2}
For two candidates with arbitrary scoring rules under constructive control, target candidate $c_1$ wins the election iff $c_1$ wins at least half of the voters. We claim that there exists a point that can be computed in polynomial time which lets $c_1$ win maximum number of voters possible within $\epsilon$. The key idea is to move $c_{1,j}$ as close to $c_{2,j}$ as possible for each issue $j\in[d]$, leaving minimum gaps. Given the voter set $\mathcal{V}$, $c_1$, $c_2$ and $\epsilon$, the algorithm is as below:
\begin{center}
    \line(1,0){239}\\
    $\tilde{c}_{1,j} = \sign(c_{2,j}-c_{1,j})\cdot \min(|c_{2,j}-c_{1,j}|,\epsilon) + c_{1,j}, j \in [d]$.\\
    If $\tilde{c}_1$ wins the election, return $\tilde{c}_1$; else return NO.\\
    \line(1,0){239}
\end{center}

Next we give a formal proof of the algorithm. Since under $l_\infty$ norm only the relative distance between $c_{1,j}$ and $c_{2,j}$, $j\in[d]$ matters. Without loss of generality,  we assume $c_1 = [0,0,\ldots,0]$, $c_2 = [a_1,a_2,\ldots,a_d]$. Since the issues are independent from one another, we permute the issues so that $0 < |a_1| \leq |a_2|\leq \ldots \leq |a_d|$ (for issues that $c_1$ and $c_2$ agrees on, i.e., $a_j =0$, we can simply omit them). Given budget $\epsilon$, we assume $|a_i| \leq \epsilon < |a_{i+1}|$. We claim that $\tilde{c}_1 = [a_1,\ldots,a_i,\sign(a_{i+1})\epsilon,\ldots,\sign(a_d)\epsilon]$ wins maximum number of voters possible within budget.

Let $\tilde{c}'_{1} = [\alpha_1,\alpha_2,\ldots,\alpha_d]$ be a point within $\epsilon$ from $c_1$ that wins maximum number of voters possible. For any voter $v_r \in \mathcal{V}$ that votes for $\tilde{c}'_{1}$, assume $\max(v_r,c_2) = |v_{r,k} - a_k|$. To start with, we set $\tilde{c}_1 = \tilde{c}'_{1}$, then we have $\max(v_r,\tilde{c}_1) = \max ({|v_{r,1}-\alpha_1|},\ldots,{|v_{r,d}-\alpha_d|)}$. Notice that since $|v_{r,k} - a_k|\geq |v_{r,j} - a_j|$, $\forall j \in[d]$, and $|a_j|\leq \epsilon$, $\forall j \in[i]$, we can set $\tilde{c}_{1,j} = a_{j}$ for $ j \in[i]$ and this has no impact on $\tilde{c}_1$ winning voter $v_r$. Now $\tilde{c}_1 = [a_1,\ldots,a_i,\alpha_{i+1},\ldots,\alpha_d]$. Next we only need to discuss the values of $\tilde{c}_{1,j}$ for ${j\in\{i+1,\ldots,d\}}$. We show that $|v_{r,k} - a_k| \geq |v_{r,j} -\sign(a_j)\epsilon|$ for all ${j\in\{i+1,\ldots,d\}}$. Due to the budget constraint to $\tilde{c}'_{1}$, we must have ${-\epsilon \leq \alpha_j \leq \epsilon}, \forall j \in [d]$. For $\forall j\in\{i+1,\ldots,d\}$ we have:
\begin{itemize}
    \item if $a_j >0$:
\begin{itemize}
    \item if $v_{r,j}\geq \epsilon \geq \alpha_j$, then $|v_{r,k} - a_k| \geq |v_{r,j} -\alpha_j| = {v_{r,j} -\alpha_j }\geq v_{r,j} -\epsilon ={ |v_{r,j} -\epsilon|}$;
    \item if $\epsilon \geq v_{r,j} \geq \alpha_j$ or $\epsilon \geq \alpha_j \geq v_{r,j} $, then since ${a_j > \epsilon \geq v_{r,j}}$, we have $|v_{r,k} - a_k| \geq |v_{r,j}-a_j| = a_j - v_{r,j} > \epsilon -v_{r,j}= |v_{r,j} -\epsilon|$.
\end{itemize}
    \item if $a_j <0$: 
    \begin{itemize}
        \item if $\alpha_j\geq -\epsilon \geq  v_{r,j}$, then $|v_{r,k} - a_k| \geq |v_{r,j} -\alpha_j| = {\alpha_j - v_{r,j}} \geq  -\epsilon-v_{r,j} = |v_{r,j} +\epsilon|$;
        \item if $\alpha_j\geq v_{r,j} \geq-\epsilon$ or $v_{r,j} \geq \alpha_j\geq -\epsilon$, then since $v_{r,j} \geq -\epsilon > a_j$, we have $|v_{r,k} - a_k| \geq |v_{r,j}-a_j| = v_{r,j}-a_j > v_{r,j} + \epsilon= |v_{r,j} +\epsilon|$.
    \end{itemize}
\end{itemize}

This means %
${\tilde{c}_1 = [a_1,\ldots,a_i,\sign(a_{i+1})\epsilon,\ldots,\sign(a_d)\epsilon]}$ wins voter $v_r$. Since $\tilde{c}'_1$ wins maximum number of voters possible, $\tilde{c}_1$ also wins maximum number of voters possible. Next we only need to check whether $\tilde{c}_1$ wins the election (i.e., wins at least half of the voters).

The computation of $\tilde{c}_1$ takes $O(d)$ time;  checking whether $\tilde{c}_1$ wins the election takes $O(m d)$ time. The algorithm takes $O(m d)$ time in total. 

\subsection{Proof of Theorem~\ref{T:rvpm_linf_cd}}

We solve the problem by finding the set which contains a representative solution point for all scenarios of $(t_1,\ldots,t_m)$ (defined as in Theorem~\ref{T:rvpm_c_ci}), which covers all the scenarios of $c_1$ winning. Each scenario represents finding a $\tilde{c}_1$ within budget constraint that gets a score of at least $f(t_j)$ from voter $v_j, j \in [m]$, and is to solve the below feasibility problem:
\begin{subequations}
\begin{align}
    &||\tilde{c}_1 - c_1||_\infty \leq \epsilon\\
    &||\tilde{c}_1 - v_j||_\infty \leq d_j^{t_j}, \quad j\in[m]
\end{align}
\end{subequations}
where $d_{j}^{t_j}$ is defined in the constructive control variant of \textsc{BVPM}.  There are in total $m\cdot |f_{\text{uniq}}|+1$ hypercubes involved for all the scenarios. We hope to find the representative points of intersections for those $m\cdot |f_{\text{uniq}}|+1$ hypercubes which then cover all the scenarios of ${c}_1$ winning. As discussed in \citeauthor{Crama95}~\shortcite{Crama95}, the problem goes down to box intersection problem for constant $d$. For a problem that involves $n$ boxes, it can be solved in time $O(n\log n)$ for $d\leq 2$ and $O(n^{d-1})$ for $d\geq 3$ \cite{Imai83,Lee83}. The resulting algorithm is exponential only in $d$.

\subsection{Proof of Lemma~\ref{avoid_ball}}
We denote $S_c$ as the boundary of $B_c$, and $S_j$ as the boundary of both closed ball $B_j$ and open ball $\mathring{B}_j$, $j\in[k]$.

As defined in Theorem 5 of \citeauthor{Crama95}~\shortcite{Crama95}, given $\{B_c,B_1,\ldots,B_k\}$ in $\mathbb{R}^d$, let $P$ be a set of points in $\mathbb{R}^d$ where for each $F\subseteq \{B_c,B_1,\ldots,B_k\}$ with $|F|\leq d$,
\begin{enumerate}[label=(\roman*)]
    \item if $\bigcap_{j \in F} S_j$ is connected, then $P$ contains a point of $\bigcap_{j \in F} S_j$, and
    \item if $\bigcap_{j \in F} S_j$ contains at most two points, then $P$ contains $\bigcap_{j \in F} S_j$.
\end{enumerate}

Assume the solution set $P'\neq \emptyset$. We chose $x\in P'$ as a representative solution. Let ${H = \{j \in \{1,\ldots,k\}\mid x\in \mathring{B}_{{j}}\}}$, $F = \{j \in \{c,1,\ldots,k\}\mid x \in S_j\}$. Let $x$ be chosen so that first $|H|$ is as small as possible, then $|F|$ is as large as possible. According to Lemma 5 in \citeauthor{Crama95}~\shortcite{Crama95}, the intersections of spheres are either connected or contain at most two points. There are four cases to discuss:

Case 1: $\bigcap_{j\in F} S_j$ is connected, and $|F| < d$. By construction of $P$, there exists a point $u \in P \cap (\bigcap_{j\in F} S_j)$. We next claim that $u \in B_c$ and for all $j \in \{1,\ldots,k\}\backslash H, u \notin \mathring{B}_j$. Since $x$ is a solution, we have $\{i_1,\ldots,i_r\}\cap H = \emptyset$, this also implies $u \notin \mathring{B}_{j}, j \in \{i_1,\ldots,i_r\}$. Since $\bigcap_{j\in F} S_j$ is connected and both $x, u \in \bigcap_{j\in F} S_j$, there is a path from $x$ to $u$ on $\bigcap_{j\in F} S_j$. If the claim is not valid, then moving from $x$ to $u$ we must encounter a first boundary of the ball $S_j$ with index $j \in \{c,1,\ldots,k\}\backslash F$. Let $v$ be the point that the path $x$ to $u$ intersects with $S_j$. If $j\in H$, this contradicts the minimality of $H$, since $v$ is in fewer balls than $x$; if $j\notin H$, this contradicts the maximality of $F$, since $v$ is on more spheres than $x$.

Case 2: $\bigcap_{j\in F} S_j$ is connected, and $|F| \geq d$. According to Lemma 5 in \citeauthor{Crama95}~\shortcite{Crama95}, there exists $F', |F'|< d$, such that $\bigcap_{j\in F} S_j = \bigcap_{j\in F'} S_j$. This goes to case 1.

Case 3: $\bigcap_{j\in F} S_j$ contains at most two points, $|F| \leq d$. By construction, $P$ contains $\bigcap_{j\in F} S_j$, and $x\in \bigcap_{j\in F} S_j \subseteq P$.

Case 4: $\bigcap_{j\in F} S_j$ contains at most two points, $|F| > d$. We take $F'\subseteq F, |F'|=d$. Since the intersections of $d$ non-coinciding spheres have at most two points, by construction $P$ contains $\bigcap_{j\in F'} S_j$, and $x \in \bigcap_{j\in F} S_j \subseteq \bigcap_{j\in F'} S_j \subseteq P$.

The proof is modified based on Lemma 3 in \citeauthor{Crama97}~\shortcite{Crama97}, where it studies the avoid ball problem in $\mathbb{R}^d$ for constant $d$ and the feasibility area is bounded by a $d$-dimensional hypercube.

\subsection{Proof of Lemma~\ref{non_convex}}

We claim that if $\tilde{y}\in \mathbb{R}^d$ satisfy constraint (\ref{budget}) and $\cup_{j=1}^d S_j(\tilde{y}_j) = [k]$, then $\tilde{y}$ is a solution to the feasibility problem. If the claim is not true, assume the $l$-th constraint of (\ref{cube_c}) is not satisfied for some $l\in[k]$, meaning $||\tilde{y}-a_l||_\infty <b_l$, this implies ${|\tilde{y}_j -a_{l,j}|<b_l}$, $\forall j\in[d]$ and $l\neq S_j(\tilde{y}_j),\forall j \in [d]$, which contradicts $\cup_{j=1}^d S_j(\tilde{y}_j) = [k]$.

Given the above result, we next show for a solution $\tilde{y}$, its $j$-th dimension coordinate $\tilde{y}_j\in[-\epsilon +y_j,\epsilon +y_j]$ can be represented by one of the points in $P_j$. Notice that for any $i$-th constraint in (\ref{cube_c}), ${i\in[k]}$, its $j$-th dimension ${|\tilde{y}_j-a_{i,j}| < b_i}$ depicts an open interval on $\mathbb{R}$: ${-b_i +a_{i,j} <\tilde{y}_j<b_i +a_{i,j}}$. The \mbox{$j$-th} dimension of $k$ constraints correspond to $k$ open intervals. Since the intervals are open, $\tilde{y}_j$ and its nearest open interval endpoint ${\tilde{y}'_j\in[-\epsilon +y_j,\epsilon +y_j]}$ satisfy ${S_j(\tilde{y}_j)=S_j(\tilde{y}'_j)}$. If there is no open interval endpoint within ${[-\epsilon +y_j,\epsilon +y_j]}$, then any $\tilde{y}'_j$ that is within ${[-\epsilon +y_j,\epsilon +y_j]}$ satisfies $S_j(\tilde{y}_j)=S_j(\tilde{y}'_j)$  (e.g., ${\tilde{y}'_j = y_j)}$. This means ${\cup_{i=1}^j S_i(\tilde{y}_i)\cup S_j(\tilde{y}'_j) \cup_{i=j+1}^d S_i(\tilde{y}_i) = [k]}$ and ${[\tilde{y}_1,\ldots,\tilde{y}_{j-1},\tilde{y}'_j,\tilde{y}_{j+1},\ldots,\tilde{y}_d]}$ satisfies constraint (\ref{budget}), it is a solution to the feasibility problem. Since $P_j$ contains all such $\tilde{y}'_j, \forall j \in [d]$, if the problem is feasible then ${P=\{p\in \mathbb{R}^d\mid p_j\in {P}_j, j \in [d]\}}$ contains a solution.

\section{Algorithm for Destructive Control in Binary Issues with a Constant Number of Voters}
\label{S:supp_binary_destructive}

For destructive control, we change the sign ``$\leq$" in constraint (\ref{dist}) to ``$\geq$". Similar to constructive control, given an arbitrary scoring function $f$ that has $r$ unique values (${|f_{\text{uniq}}|=r}$), since $f$ is non-increasing, we can partition the domain of $f$ by ${1 = s_1 < s_2 < \dots < s_r < s_{r+1} = n+1}$, so that $\{f(k)\}_{k=s_i}^{s_{i+1}-1}$ have the same value, $i \in \{1,2, \ldots,r\}$. This means $s_{i}$ is the highest ranking position that corresponds to score $f(s_{i})$ and $\{f(s_i)\}_{i=1}^r$ contains all the unique values of $f$.

Next we solve the problem for each $(t_1,\ldots,t_m)$ scenario with $t_j \in\{s_1,\ldots,s_r\}, \forall j \in [m]$. For each voter $v_j, {j \in[m]}$, we rank candidates $c_i ~(i \geq2)$ by their distances to voter $v_j$ from closest to furthest, and use $d_j^{t_j}$ to denote the distance between $v_j$ and the candidate ranked \mbox{$(t_{j}-1)$-th} closest to it. Since the tie breaks in the adversary's favor, as long as $c_1$'s distance to $v_j$ is at least $d_j^{t_j}$, $c_1$ will receive a score of no more than $f(t_j)$ from $v_j$. Notice that since the rankings of $c_i~(i\geq 2)$ do not include $c_1$, only $d_j^{t_j}$ for $2\leq t_j\leq n$ are properly defined. For $t_{j} = s_1= 1$, we let $d_j^{1} = 0$, since $c_1$ is guaranteed to get no more than the highest score $f(1)$. 

The correctness for the above arguments for destructive control can be demonstrated by the below two lemmas similar to constructive control.

\begin{lemma}
For destructive control, if a ranking position $(r_1,r_2,\ldots,r_m)$ is feasible for ${c}_1$ under the budget constraint and lets ${c}_1$ lose the election, then for a ranking position $(r'_1,r'_2,\ldots,r'_m)$ that is feasible with $r'_j\geq r_j, \forall j\in [m]$, it will also let ${c}_1$ lose the election.
\end{lemma}
\begin{proof}
For each voter $v_j, j \in [m]$, we rank candidates $c_i ~({i\geq 2})$ by their distances to $v_j$ from closest to furthest. %
Without loss of generality, we assume candidates $c_n,c_{n-1},\ldots,c_2$ are ranked from closest to furthest w.r.t. $v_j$. Then when $c_1$ is inserted into this sequence and is ranked $r_j$, the score each candidate receives from $v_j$ is as below:

\begin{table}[H]
\centering
\begin{tabular}{|c|c|c|c|c|c|c|}
\hline
$c_{n}$ & ... & $c_{n-r_j+2}$   & $c_1$ & $c_{n-r_j+1}$     & ... & $c_{2}$ \\
\hline
$f(1)$ & ... & $f(r_j-1)$ & $f(r_j)$ & $f(r_j+1)$ &  ...  & $f(n)$ \\
\hline
\end{tabular}
\end{table}

If we move $c_1$ from ranking position $r_j$ to $r'_j$, ${r'_j \geq r_j}$, since $f$ is a non-increasing function, the score $c_1$ receives from $v_j$ will decrease from $f(r_j)$ to $f(r'_j)$; the scores candidates $c_{n},c_{n-1},\ldots,c_{n-r_j+2}$ and $c_{n-r'_j+1},c_{n-r'_j},\ldots,c_{2}$ receive from $v_j$ will not change; for candidates $c_{n-r_j+1},c_{n-r_j},\ldots,c_{n-r'_j+2}$, since their rankings will increase by $1$, the scores they receive from $v_j$ will not decrease. This means if we move $c_1$ from ranking position $(r_1,r_2,\ldots,r_m)$ to $(r'_1,r'_2,\ldots,r'_m)$, ${r'_j \geq r_j}, \forall {j \in [m]}$, the total score $c_1$ receives will not increase, while the total scores other candidates receive will not decrease. Since $c_1$ loses the election with ranking position $(r_1,r_2,\ldots,r_m)$, $c_1$ will also lose the election with ranking position $(r'_1,r'_2,\ldots,r'_m)$.
\end{proof}

\begin{lemma}
For destructive control, by enumerating all scenarios of ${c}_1$ getting a ranking position lower than $t_j$ w.r.t. voter $v_j$ for $t_j \in\{s_1,\ldots,s_r\},\forall j\in [m]$, we cover all the possible scenarios of ${c}_1$ losing.
\end{lemma}
\begin{proof}
Without loss of generality, we assume candidates $c_n,c_{n-1},\ldots,c_2$ are ranked from closest to furthest w.r.t. voter $v_j$. We next insert $c_1$ into this sequence and assume $c_1$ receives a score of $f(s_{i})$ from $v_j$, where $i\in\{1,\ldots,r\}$. After the insertion, the score each candidate receives from $v_j$ is as below:

\begin{table}[H]
\centering
\begin{tabular}{c||c||c||c}
\hline
$...,c_{n-s_i+2}$ & $... ,c_1, ... $ & $ c_{n-s_{i+1}+2}, ...,c_{n-s_{i+2}+1}$ &\\
\hline
$f(s_{i-1})$ & $f(s_{i})$ & $f(s_{i+1})$ &\\
\hline
\end{tabular}
\end{table}

Notice that whichever ranking position $c_1$ takes among $s_i ,{s_i +1},\ldots,s_{i+1}-1$, the final score each candidate receives from $v_j$ are the same, meaning these ranking positions are equivalent. 
Since the solution set of getting a ranking position lower than $s_{i}$ covers the solution set of getting a ranking position lower than $s_i +1,s_i +2,\ldots,s_{i+1}-1$, we only need to consider $s_{i}$.
\end{proof}

\section{Algorithm in Lemma \ref{non_convex_poly}}
\label{algo_1_des}

\begin{algorithm}
\caption{Feasibility problem with constant number of constraints}
\renewcommand{\algorithmicrequire}{\textbf{Input:}}
\renewcommand{\algorithmicensure}{\textbf{Output:}}
\begin{algorithmic}[1]
\REQUIRE Feasibility problem parameters.
\ENSURE A solution to the problem; or NO if not feasible.
\STATE Calculate $\mathcal{S}_1$ and $P_1$ \label{init}
\STATE $\mathcal{S}\leftarrow\mathcal{S}_1$, $P\leftarrow P_1$ \label{init_sets}
\FOR{$j \in [2:d]$}
\STATE Calculate $\mathcal{S}_j$ and $P_j$
\STATE Find all pairs of sets $S(p)\in \mathcal{S}$ and $S_j(p_j)\in \mathcal{S}_j$, where $S_j(p_j) \nsubseteq S(p)$. %
\STATE Add $S(p) \cup S_j(p_j)$ to $\mathcal{S}$ and $[p,p_j]$ to ${P}$
\IF{$[k]\in \mathcal{S}$} \label{terminate_k}
\STATE Expand its corresponding point $p$ to $d$-dimension
\RETURN $p$
\ENDIF
\STATE Remove subsets in $\mathcal{S}$ and corresponding points in $P$ \label{remove_subset}
\STATE Update $P$ so that all points have dimension $j$ \label{update_p_j}
\ENDFOR
\RETURN NO
\end{algorithmic}
\label{feas_poly}
\end{algorithm}

For simplicity, we first define set ${\mathcal{S}_j = \{S_j(p_j)\mid p_j \in P_j\}}$, where $P_j$ and $S_j(p_j)$ are defined in Lemma \ref{non_convex}. Assume a point $p$ is $j$-dimensional, we let $S(p) = \cup_{i = 1}^j S_i(p_i)$. According to Lemma \ref{non_convex}, finding a solution to the feasibility problem is equivalent to finding a point $p\in\mathbb{R}^d$ that satisfy constraint (\ref{budget}) and $S(p)=[k]$.

In Line \ref{init} and \ref{init_sets}, we initialize the sets $\mathcal{S}$ and $P$ with $\mathcal{S}_1$ and $P_1$, and later use them to store the results for the $[1:j]$ dimension. 
In the main loop, for each dimension $j\in [2:d]$, we first calculate sets $\mathcal{S}_j$ and $P_j$, and then find all pairs of sets $S(p)\in \mathcal{S}$ and $S_j(p_j)\in \mathcal{S}_j$ that satisfy $S_j(p_j) \nsubseteq S(p)$. This means $S(p)\subset S([p,p_j])$. We add all such sets $S([p,p_j])$ to $\mathcal{S}$ and points $[p,p_j]$ to $P$. The terminate condition is if $[k]\in \mathcal{S}$, then for the $j$-dimensional point $p\in P$ that has $S(p)=[k]$, $S([p,y_{j+1},\ldots,y_d])=[k]$ and $[p,y_{j+1},\ldots,y_d]$ satisfies constraint (\ref{budget}). $[p,y_{j+1},\ldots,y_d]$ is a solution to the feasibility problem.

In Line \ref{remove_subset} we remove all subsets in $\mathcal{S}$ (as well as the corresponding points in $P$) and keep the sets in $\mathcal{S}$ pairwise incomparable. This improves our algorithm efficiency and has no impact on us finding a solution: %
if a solution $\tilde{y}$ exists with ${\cup_{i=1}^d S_i(y_i) = [k]}$, and we have $y'_j \in {[-\epsilon +y_j,\epsilon +y_j]}$ with ${S_j(y_j) \subseteq S_j(y'_j)}$, we then have ${\cup_{i=1}^{j-1} S_i(y_i) \cup S_j(y'_j) \cup_{i =j+1}^d S_i(y_i) = [k]}$, meaning $[y_1,\ldots,y_{j-1},y'_j,y_{j+1},\ldots,y_d]$ is also a solution. 

In Line \ref{update_p_j}, we append $y_j$ to all the $(j-1)$-dimensional points in $P$ so that all points are $j$-dimensional.

Since $k$ is constant, the sizes of $\mathcal{S}$ and $\mathcal{S}_j$ as well as their elements are all constants. The algorithm is linear and of complexity $O(d)$. We can also determine in linear time if none of the representative points satisfy the feasibility condition, and return NO in that case.
\end{document}